\documentclass[letterpaper,10pt]{article}

\usepackage[lined, linesnumbered, ruled]{algorithm2e}

\pdfoutput=1

\usepackage[
activate={true,nocompatibility}, 
final, 			
tracking=true,
kerning=true,
spacing=true,
factor=1100, 	
stretch=20, 	
shrink=20 		
]{microtype}

\usepackage{xcolor}

\usepackage[utf8]{inputenc}
\usepackage[
backend=bibtex,		
natbib=true, 		
style=authoryear-icomp, 
doi=false,			
url=false,			
hyperref=true,		
backref=true, 		
maxnames=1, 		
maxbibnames=100 	
]{biblatex}
\usepackage{graphicx}
\usepackage[letterpaper, margin=1in]{geometry}

\usepackage[full]{textcomp}
\usepackage[cal=boondoxo]{mathalfa} 

\usepackage[labelfont=bf, textfont={rm,it}, margin=0.5in]{caption}
\usepackage{subcaption} 

\usepackage{bm} 
\usepackage{amsmath, mathrsfs,amsfonts,amssymb}
\usepackage{wasysym} 
\usepackage{xfrac} 

\usepackage[retainorgcmds]{IEEEtrantools} 

\usepackage{savesym} 
\savesymbol{openbox}
\usepackage[amsthm,amsmath,thmmarks]{ntheorem}

\theoremstyle{plain} 		\newtheorem{theorem}{Theorem}[section]
\theoremstyle{definition} 	
\theoremstyle{remark}		
\newtheorem{lemma}[theorem]{Lemma}

\usepackage{booktabs}
\usepackage{units}

\usepackage{fancyhdr} 
\DeclareMathOperator*{\argmax}{arg\,max}

\usepackage{authblk} 		
\newcommand{\citeN}{\textcite} 
\renewcommand{\cite}{\parencite} 

\begin{document}
\title{ \Large \bf Influencing Busy People in a Social Network}
\author[1]{\bf Kaushik Sarkar}
\author[2]{\bf Hari Sundaram}
\affil[1]{School of Computing, Informatics and Decision Systems Engineering}
\affil[ ]{Arizona State University, Tempe, Arizona 85281}
\affil[2]{Dept. of Computer Science}
\affil[ ]{University of Illinois at Urbana-Champaign, Urbana, Illinois 61801.}
\affil[1]{\texttt{ksarkar1@asu.edu}}
\affil[2]{\texttt{hs1@illinois.edu}}

\renewcommand\Authands{ and }


\date{\today}
\maketitle

\begin{abstract}
We identify influential early adopters in a social network, where individuals are resource constrained, to maximize the spread of multiple, costly behaviors.  A solution to this problem is especially important for viral marketing. The problem of maximizing influence in a social network is challenging since it is computationally intractable. We make three contributions. First,  propose a new model of collective behavior that incorporates individual intent, knowledge of neighbors actions and resource constraints. Second, we show that the multiple behavior influence maximization is NP-hard. Furthermore, we show that the problem is submodular, implying the existence of a greedy solution that approximates the optimal solution to within a constant. However, since the greedy algorithm is expensive for large networks, we propose efficient heuristics to identify the influential individuals, including heuristics to assign behaviors to the different early adopters. We test our approach on synthetic and real-world topologies with excellent results. We evaluate the effectiveness under three metrics: unique number of participants, total number of active behaviors and network resource utilization.  Our heuristics produce 15-51\% increase in expected resource utilization over the na\"ive approach.

\end{abstract}

\section{Introduction}
Often, we are unable to adopt a behavior---despite knowledge of behavior utility, interest in adoption and knowledge of behavior adoption amongst friends---because we lack the resources to adopt the behavior. Consider the following scenarios: Many of John's friends are planning to take the free flu shot offered at the university health clinic this week. Unfortunately for John, he is planning to be away from the university at that time and unable to make use of the opportunity. Mary wants to buy and cook with organic produce like many of her friends, but in her small town the only grocer that sells organic produce is ten miles away; not owning a car, she cannot afford to spend time traveling to the store and back on a public bus. Both examples point to absence of resources, either time, or a tangible resource like a car being an important barrier to adoption of behavior.

This paper investigates the problem of how to maximize the spread of multiple, costly behaviors in a social network when individuals have limited resources. In general, resources can be time, money or any tangible asset. We plan to address this problem by identifying a small set of influential individuals, who by becoming early adopters of the behavior will maximize the spread of the behavior in the network.  The ability to identify such individuals is vital to viral marketing. The problem of \textit{multiple} behavior ``influence maximization'' can be shown to be computationally intractable. Indeed, a simpler version of influence maximization, when we are interested in maximizing the spread of a \textit{single} behavior, \textit{without resource constraints} has been shown to be NP-hard~\parencite{Kempe03}.

Our research in the more generalized influence maximization problem is motivated by a desire to model and analyze a more realistic decision making process: individuals have to decide which subset of behaviors---amongst those that their friends have adopted---to adopt that consume fewer resources than what they possess. Faced with such a scenario, we would expect a rational individual to adopt that subset of behaviors that maximizes their utility while subject to resource constraints. That individuals with limited resources affect the behavior dynamics has empirical support:~\citeN{Hodas2012} suggests that resource constraints help explain why information stops spreading on networks like Twitter.  Resource constraints not only limit individual participation, but also shape how behaviors spread in a network. Although there is some work on diffusion of two competitive behaviors, the general idea of individual level resource constraint, and the presence of multiple costly behaviors is largely unexplored in past work.

In this paper, we develop a model of multiple behavior diffusion that captures the complex dynamics of multiple behavior adoption in resource constrained networks. We answer two specific questions with this model: who are the $k$ most influential individuals in a network? and what set of behaviors should these influential individuals adopt, if they are chosen as the early adopters? We make several contributions in this paper. 
\begin{itemize}
\item  To the best our knowledge, our work is the earliest of its kind to analyze the influence of individual resource constraints on multiple, costly behavior adoption in an network. We propose a model of multiple behavior adoption that extends the earlier threshold model developed by \textcite{Granovetter1978}, via incorporating resource constraints and individual intent. Thus, an individual adopts a behavior when she receives a ``social signal''---knowledge of neighbors actions---of sufficient strength (that is, when sufficient number of neighbors have adopted), the behavior is of high utility and when she has the resources to do so. 
\item  We show that the influence maximization problem in the multiple behavior case is NP-hard. However, we note that this result is  unsurprising  since the influence maximization problem in the single behavior case which is a special case of the research question addressed in this paper was shown to be NP-hard by~\textcite{Kempe03}. Importantly, we prove that the influence maximization problem for multiple behaviors,  with individual resource constraints is \textit{submodular},  implying the existence of a greedy algorithm that approximates the optimal solution within a factor of $1-1/e$. 
\item We propose several efficient heuristics to approximately solve the NP-hard influence maximization problem, including identifying the behavior that each seed ought to adopt. For example, our Expected Immediate Adoption heuristic is $O(n \log n + m)$, where $n$ is  the number of nodes in the graph and $m$ is the number of edges. In contrast, the greedy algorithm that best approximates the optimal solution is $O(n^2)$ with a large constant since the greedy algorithm requires costly stochastic simulation to evaluate the diffusion. The heuristics to assign behavior include assigning behaviors randomly to each node, assigning behaviors proportional to cost, and assigning all seed nodes to have the highest cost behavior.
\end{itemize}

We use three different metrics to evaluate the different heuristics: unique number of participants, number of behaviors in the network and expected resource utilization.  We test our approach on synthetic and real-world topologies with excellent results. We show that two heuristics that evaluate Influence Weight and Expected Immediate Adoption provide very good solutions to the seed selection problem. Our seed selection heuristics produce 15-51\% increase in expected resource utilization over the na\"ive approach of identifying individuals with the highest degree as influential. We find that when behaviors are assigned to seeds proportional to behavior cost, we have the highest resource utilization, and when all seeds are assigned the lowest cost behavior, we see the highest participation.

The rest of the paper is organized as follows. In the next section we review the relevant literature. In Section~\ref{sec:model} we formally define our behavior diffusion model. In Section~\ref{sec:seed-selection} we define the seed selection problem, prove intractability results, and provide an approximation algorithm for a slightly simplified model. In section~\ref{sec:heuristics} we present different heuristics and compare their performance for synthetic as well as real world networks. In Section~\ref{sec:dist} we discuss different behavior distribution strategies and present simulation results.
In Section~\ref{sec:disc} we discuss open issues, extensions and our conclusions. Appendix illustrates and elaborates a few technical issues that are not addressed fully in the main paper.

\section{Related Work}
There is a rich body of literature spanning multiple disciplines that have analyzed the problem of diffusion of behavior.  It would be infeasible to present an exhaustive survey of past work related to this paper and we hope to present a representative sample here.  

Motivated by the work of~\citeN{rogers62} amongst others,~\citeN{bass69} proposed a temporal model of sales of infrequently bought consumer durables (e.g. buying refrigerators). The model proposed that the probability of initial purchase at a time $t$, given that no purchase has been yet made is a linear function of previous buyers.  Based on earlier work that identified different buying behaviors,~\citeN{bass69} proposed that there were essentially two types of consumer behavior---\textit{innovators} and \textit{imitators}. Innovators bought products without being influenced by social pressures, whereas imitators were influenced by the adoption decisions of other buyers. While Bass model has been highly influential---the result is a simple model of aggregate behavior whose parameters can be estimated from sales data---it is a population model that ignores the network structure in which an individual finds herself.~\citeN{Granovetter1978} proposed a simple social influence model of adoption: an individual adopts a behavior if a certain fraction of the group adopt the behavior, and where the fraction exceeds the individual's private threshold.  

As~\citeN{Young2009} points out in his study of innovation diffusion, a major limitation of threshold models is that we do not know \textit{why} an individual is influenced by his peers since Granovetter's model lacks an economic incentive for the individual to adopt the behavior.  This can be addressed via a threshold model that arises out of a network coordination game: an individual adopts a behavior to coordinate with her network neighbors in a manner that maximizes her utility. It is easy to show that the network coordination game, where each neighbor has adopted behaviors with different utilities, is equivalent to each individual possessing a local, private threshold that must be exceeded for her to adopt.

\citeN{Watts02} analyzes random graphs using the linear threshold model to identify conditions for the emergence of global cascades. He find that if ``vulnerable'' nodes percolate\footnote{A node is ``vulnerable'' if its threshold $\theta \leq 1/K$ where it has $K$ neighbors. The vulnerable nodes are said to percolate when they form a giant component.} then global cascades can occur.  In more recent work~\citeN{Watts07} critically evaluate the ``influential hypothesis" which has played a significant role in the development of the theory of social diffusion processes. The influential hypothesis posits that opinion leaders or ``influentials" act as intermediaries in the dissemination of the information from mass media to general public. \citeN{Watts07} shows that except for some special situations large scale social cascades are driven by a critical mass of easily influenced individual rather than the influentials. However, mechanisms to efficiently trigger cascades in such networks is unclear.

There has been much work in Computer Science on the problem of influence maximization---how to efficiently identify seeds (or ``innovators'' in the parlance of~\cite{bass69}) that maximally influence the network.~\citeN{Domingos02} and~\citeN{Kempe03} initiated the study of the computational problem of seed selection in the context of a ``viral'' social diffusion process.~\citeN{Kempe03} formalized the algorithmic problem for \textit{Independent Cascade} and \textit{Linear Threshold} models, proved the intractability results and provided a greedy approximation algorithm based on earlier work by~\citeN{Nemhauser78}. However the greedy algorithm is computationally very expensive in practice since identifying each seed requires a large number of stochastic simulations, Much of the recent work(e.g.~\cite{Leskovec07}) has focused on reducing the computational complexity of the simulations. Identifying cheap computational heuristics that match the performance of~\citeN{Kempe03} are surprisingly less common.~\citeN{Chen09} who developed computationally cheap heuristics for the Independent Cascade model\footnote{In the Independent Cascade model, a vertex $v$ has a fixed probability $p$ of influencing each of its network neighbors to adopt the behavior.  Crucially in the model, it can attempt to influence each of its neighbors just once with probability $p$. There has been much work to identify these probabilities for the independent cascade model~\cite{Saito08,Goyal10,Mathioudakis11}.} that matches the performance of the greedy approximation algorithm is an exception. 

Motivated by earlier work in Economics on cascades by~\citeN{Arthur1989},~\citeN{Chierichetti2014} and~\citeN{Martin2014} study the scheduling of cascades on a arbitrary graph. In the problem that~\citeN{Chierichetti2014} study, there are two competing products and individuals choose one product over the other keeping in mind two factors: their own preference for each product and the payoff from aligning with the choices of their neighbors. If one product has a preference probability $p$ the other has preference probability $1-p$. They show that if they can \textit{schedule} the order in which individuals make decisions on which of the two competing products that they adopt, the number of adoptions is linear in the size of the social network. There is an important difference between the~\citeN{Arthur1989} model and the linear threshold model adopted in this paper. The difference is that while in the~\citeN{Arthur1989} model, a person examines the numbers of adoptees for \textit{both} products to make a decision, in the linear threshold model, a person makes a decision on adopting a product \textit{only} based on number of adoptees for that product. Due to this distinction, the number of adoptees for a product in the linear threshold model is monotone over time and the outcome is order independent.

\citeN{Seeman13} examines the adaptive seeding problem in the framework of two-stage stochastic optimization. The problem here is: given a seed budget, and a subset $X$ of accessible users, how to choose a set of seeds from $X$, utilizing only part of the seed budget, such that the expected value of the influence function (i.e. spread) can be maximized by utilizing the remaining seed budget on some subset of their neighbors. \citeN{Rubinstein15} analyzes the same problem in the more general setting, where different individuals have different activation costs. Although they call this setting by the name ``knapsack" constraint, it differs from the knapsack constrain that we have imposed in our model (see section \ref{sec:model}) in two important aspect: first, they do not consider multiple competing behaviors, and second, in their setting it is more appropriate to consider that the advertiser or the campaign runner is solving a knapsack problem, whereas in our model each individual in solving a knapsack type of constraint for making behavior adoption decisions.  

In Economics, there is a rich body of work~\cite{Bala1998a,Golub2010,Acemoglu2011a} that examines both social learning---how individuals adopt beliefs---and consequently identifying influentials. Social learning is the idea that rational agents take optimal decisions based on the observations of other agents.  While assumptions vary---whether the network is fully observable; only adoptions but not payoffs are observable---the idea that individuals choose the behavior that maximizes their utility has a natural appeal. For example,~\citeN{Bala1998a} introduce the idea of ``learning from neighbors''---analysis of a set of infinite agents who can observe the actions and outcomes of only their network neighbors. The main result is that in a connected graph, if the actions are ranked by payoffs, then in the long run, everyone chooses the same action with probability one.  

The~\citeN{DeGroot1974} influence model has seen significant follow up work (e.g.~\cite{Golub2010}). In the basic model, a person updates their belief (e.g. about the occurrence of an event) using a weighted sum of the beliefs of their neighbors.  For a connected, aperiodic\footnote{A directed graph is aperiodic if 1 is the greatest common devisor of all the lengths of the different directed cycles.}, directed graph, we can show that all closed\footnote{A closed set of nodes is one where any node in this set is not influenced by nodes outside this set.} strongly connected components arrive at a consensus; different closed components will in general arrive at different consensus values.  A node not belonging to any closed component will arrive at a belief that is a weighted average of consensus beliefs of components to which it is connected. Thus we can compute a measure of ``influence'' in the following manner. Assume that $T$ is the influence weight matrix of the social network where $T_{i,j}$ is the weight that node $i$ has for node $j$. Then, it is straightforward to show that $s_{j}$ the influence of a node $j$ is the the $j^{\mathrm{th}}$ entry of the leading left eigenvector $s$ of $T$, where $sT = s$.   One weakness in thinking in terms of influence in the~\citeN{DeGroot1974} model is that since only the closed sets of nodes arrive at a consensus, it is possible to construct directed graphs where a small set of closed nodes have all the influence since only closed sets of nodes will have non-zero entries in the eigenvector $s$. This is exactly the issue in the original PageRank algorithm where the PageRanks of nodes except for closed sets will go to zero.  A simple solution in the PageRank case was to used a scaling with random restarts. Extending the scaled PageRank model to the~\citeN{DeGroot1974} model essentially implies that individuals look at the average beliefs of their neighbors with probability $p$ and a randomly chosen individual's belief with probability $1-p$, where $p$ is the scaling factor.  

There is a clear distinction between the linear threshold models (e.g.~\cite{Kempe03}) analyzed in Computer Science literature and the~\citeN{DeGroot1974} influence model or the Bayesian social learning frameworks (e.g.~\cite{Bala1998a}) analyzed in research papers in Economics. In the linear threshold model, the adoption is ``progressive''---once a node adopts a behavior, it will never drop that behavior~\cite{Kempe03}. In the~\citeN{DeGroot1974} influence model and the social learning frameworks, individuals update their beliefs at every time-step as a weight average of the beliefs of their network neighbors. Their beliefs will converge in the limit\footnote{It can be shown the convergence speed depends on the size of the second largest eigenvalue of $T$, the influence weight matrix. See~\citeN{Jackson2008} for a textbook proof.}.

Our work on maximizing the spread of behavior in resource constrained networks is informed by this literature, but is markedly different in a number of aspects. Much of the existing literature is concerned with diffusion of a single influence, while simultaneous diffusion of multiple influences is a more realistic scenario. The idea that individuals are resource constrained is important---were it not so, adoption rates of multiple behaviors would follow trivially from prior work on influence maximization~\parencite{Kempe03}. Introduction of resource constraints implies that each individual now has to choose from a subset of behaviors adopted by his network neighbors such that the cost of her behaviors is less than her available resource in a manner that maximizes utility. Thus, each individual solves a knapsack problem of picking behaviors to maximize their utility while these behaviors satisfy their individual resource constraints.~\citeN{Bharathi2007} and \citeN{Carnes07} discuss the problem of multiple competing influences, but they also do not incorporate the resource constraints or the utility maximizing behavior of individuals into their models.  That individuals in a social network are resource constrained has an empirical basis---\citeN{Hodas2012} provide empirical evidence in support of the hypothesis that social contagions are constrained by finite amount of resources (e.g. time to process information) available to the individuals constituting the social network. However, their work does not shed any light on the algorithmic question of influence maximization in resource constrained networks. 

Our notion of individual resource constraints is a form of \textit{bounded rationality}, a well known idea in Economics (e.g.~\parencite{Simon1972,Kahneman2003}). In our framework, individuals cannot view the entire network, but only their network neighborhood. However, we do assume that the network neighborhood is fully observable and that each individual is rational---able to evaluate the utility of a set of behaviors in her neighborhood and be able to identify the set that maximizes utility.  To the best of our knowledge the present work is the first investigation of the seed selection problem for multiple behavior diffusion in a resource constrained social network.

\section{Multiple Behavior Adoption for Resource Bounded Networks}\label{sec:model}
In this section we introduce our multiple behavior diffusion model. First, we will describe our model of multiple behavior diffusion when individuals have bounds on resources available to them. Then we will introduce metrics  including resource utilization, unique participation and number of behavior adoptions to evaluate the behavior adoption process. We represent the social network by an undirected graph $G=(V,E)$. Each node $v \in V$ of the graph $G$ represents an individual and an edge $e \in E$ between two nodes indicate a social relationship between the two individuals. Without loss of generality, we assume that the goal is to spread $k$ behaviors in the network.

\subsection{Our Diffusion Model} \label{sub-sec:model}
We now describe the model for each user, the properties of each behavior and the behavior adoption process. For easy reference table \ref{tab:glossary} presents the symbols used subsequently in this paper. Conceptually our behavior adoption model can be described as follows---an individual adopts a new behaviors if the behavior has some value to her i.e. she has some interest in the behavior (intent), a significant number of her friends have adopted the behavior (social signal), and she has enough available resource to pursue it (resource). 

\begin{table}[htb]
\footnotesize
\centering
    \caption{Glossary of symbols used in the paper. The first half contains symbols used in this section, and the second half contains symbols used in the next section on influential identification problem
    }\label{tab:glossary}
    \begin{tabular}{cp{11cm}} \toprule
        Symbol & Meaning \\ \midrule
        $G$ & The undirected graph that represents the social network \\
        $V$ & The set of all individuals in the social network \\
        $E$ & The set of all social relationships between individuals \\
        $k$ & Number of behaviors  \\
        $c_{i}$ & Cost of adoption associated with behavior $i$. $0\le c_{i} \le 1$ \\
        $u_{i}$ & Utility obtained from adoption of behavior $i$. $0\le u_{i}\le 1$ \\
        $r(v)$ & Resource available to $v\in V$ for behavior adoption. $0\le r(v) \le 1$ \\
        $N(v)$ & The set of neighbors of $v\in V$ in the social network \\
        $\theta_{i}(v)$ & Threshold associated with behavior $i$ for individual $v\in V$.
        $0\le\theta_{i}(v)\le 1$ \\
        $l_{i}(v)$ & Strength of social signal associated with behavior $i$ -- i.e. sum
        of influence weight exerted by the neighbors with adopted behavior
        $i$ -- on $v\in V$. $0\le l_{i}(v) \le 1$ \\
        $w$ & Relative weight assigned to utility for computing the individual payoff.
        $0\le w \le 1$  \\
        $p_{i}(v)$ & Payoff associated with behavior $i$ for individual $v\in V$, defined
        as $wu_{i}+(1-w)l_{i}(v)$ \\
        $s(v)$ & The amount of resource that individual $v\in V$ spends to adopt behaviors \\
        $b$ & The fixed budget (number) of seeds in the influential identification
        problem \\ \midrule
        $S$ & The set of seeds for the $k$ different behaviors \\
        $\sigma(S)$ & Expected number of individuals with at least one behavior at the end,
        starting with the seed set $S$ \\
        $\kappa(v)$ & The largest index $j$ such that $c_{j}\le r(v)$ for the individual
        $v\in V$ (We assume that the behaviors are indexed in the non-decreasing
        order of cost) \\
        $b_{v,w}$ & Influence weight exerted by the neighbor $w$ on the individual $v$.
        $w,v\in V$ \\
        $S_{i}^{(t)}$ & The set of individuals with behavior $i$ at the end of time step
        $t$ in the sticky multiple behavior case \\
        $S^{(t)}$ & The set of individuals with at least one behavior at the end of time
        step $t$ \\
        $\sigma'(S)$ & Expected number of nodes with at least one behavior at the end of
        the process, starting with seed set $S$ \\
        $R(v,X)$ & The set of nodes with behavior seeded by $v\in S$, and reachable
        from $v$ through a live edge path under the particular choice of
        live/blocked edges $X$ \\ \bottomrule
    \end{tabular}

\end{table}

Each behavior $i$ has a cost $c_i $ and a utility $u_i $ associated with it.  
In a simplification, we assume that both the cost $c_i $ and the utility $u_i $ of behavior $i$ are intrinsic to the behavior and independent of the individual who adopts the behavior. Without loss of generality, we assume that $0 \leq c_i , u_i  \leq 1$.

Individuals are resource constrained: an individual may have limited time, money or may not possess other material resources to adopt a behavior. Therefore, we assign a \textit{fixed} resource $r(v)$ for each individual $v \in V$ towards adopting behaviors. The resource satisfies $0 \leq r(v) \leq 1$. For example, if we assume that individuals' resources are independent and identically distributed then the resource value $r(v)$ can be assumed to be obtained from a uniformly distributed random variable $U(0,1)$.  Without loss of generality we assume that the resource type (e.g. money, time) is the same as the cost type. Let $N(v)$ denote the set of neighbors of $v$ in the network. Then we assume that a neighboring node $u$ asserts a social influence on node $v$ with weight $1/|N(v)|$.  

An individual will adopt a behavior $i$ when she receives a strong social signal, has the resources to do so and when there is sufficiently high payoff in adopting the behavior. A behavior is a likely candidate for adoption when the strength of social signal exceeds a threshold, and the individual has enough resource to adopt the behavior. Figure \ref{fig:concept} depicts the situation where the candidate behaviors are those for which the high social signal and resource availability conditions are met. We assume that each individual $v$ has a different, fixed, threshold $\theta_i (v)$ for each behavior, and that each threshold is obtained independently from a uniformly distributed random variable $U(0,1)$. The strength of social signal is measured by $l_i (v)$ which is defined as the sum of influence weights---the social signal---exerted on $v$ by its neighbors who have adopted behavior $i$. The payoff $p_i (v)$ for a behavior $i$ is defined as the weighted sum of the intrinsic utility $u_i $ and the local network utility $l_i (v)$. That is, $p_i (v) = wu_i  + (1 - w)l_i (v)$. Where, $w$ denotes the relative weight of the intrinsic utility (i.e. intent). Figure \ref{fig:concept} also shows this situation where the payoff is determined by social signal and intent. An individual adopts only those candidate behaviors that have high payoff (shown as the intersection between Candidate and Payoff in Figure \ref{fig:concept}). If there are multiple candidate behaviors, then an individual adopts a subset of candidate behaviors that maximizes total payoff.

\begin{figure}[htb]
    \begin{centering}
    \includegraphics[width=\columnwidth]{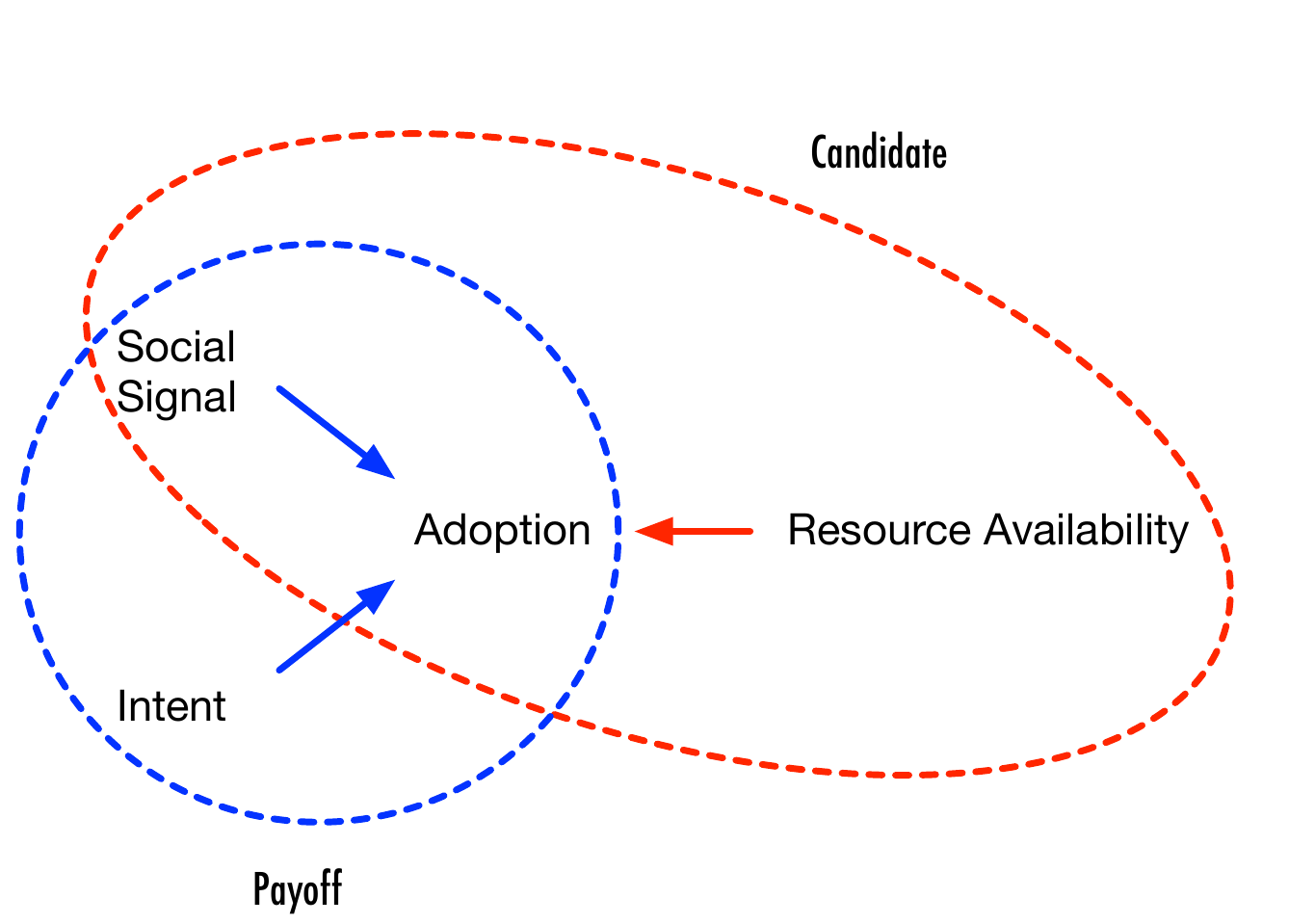}
    \par\end{centering}
    \caption{Payoff for adopting a behavior comes from an individual's intent and social signal.}\label{fig:concept}
\end{figure}

Let us examine the diffusion of behavior over time, to illuminate the key ideas. The process takes place over discrete epochs\footnote{Notice that while individual actions in a network are asynchronous, we can choose an appropriate time granularity for analysis to assume synchronized decision making.}. We assume each node is aware of the behaviors adopted by her neighbors. The individual $v$ first identifies all candidate behaviors. A behavior $j$ is a candidate to be adopted if two conditions hold. First, the social signal strength for behavior $j$ must exceed the threshold for that behavior at node $v$, i.e. $l_j(v) \ge \theta_j(v)$. Second, the individual $v$ must have the resources to adopt the behavior, i.e. $r(v) \ge c_j$.  The first condition is the familiar Linear Threshold (LT) model \parencite{Kempe03}. In problem formulation there are multiple behaviors, and the individual $v$ chooses a subset of candidate behaviors that maximizes the total payoff, subject to the condition that the sum of the adoption costs of the behaviors is less than the resource constraint. Let $B_v$ be the set of candidate behaviors for an individual $v$. So $v$ adopts a set of behaviors $B \subseteq B_v$ that maximizes $\sum_{i \in B}p_i (v)$ subject to the constraint that $\sum_{i \in B}c_i \le r(v)$. At every epoch, the individual $v$ evaluates all behaviors, including behaviors already adopted, to evaluate payoff.
The behavior diffusion process continues until no additional adoption is possible. 

In our diffusion model, we assume that the total resources available $r(v)$ at each node  are known, while the threshold for adoption $\theta$ for any behavior is unknown. This assumption is reasonable when people are willing to make public their available resources to participate in a set of behaviors. This can arise say in a private, mobile social network app focused on adoption of healthy behaviors including wellness, healthy eating and exercise, where individuals join the network to participate in healthy behaviors but each individual is resource limited. An individual may declare that she has only one hour to spend on exercise each week, but would like to be nudged to participate in a health-related activity.

Figure \ref{fig:model} shows an illustration of the spread of behaviors with a four node network where three different behaviors---recycling, using public transport and eating locally grown food---are denoted by behaviors $1$, $2$ and $3$ respectively. At time step $0$ the state of the network is shown in \ref{fig:state-before}. At this time step, for $v$, the social signal of eating locally grown food is weak. So $v$ considers only the recycling and the using public transport behaviors for adoption.  After maximizing payoff subject to the resource constraint, $v$ adopts only the recycling behavior.  Although public transport has strong social signal, $v$ cannot adopt that behavior because it does not have enough resources to adopt the behavior. Notice that the payoff for recycling is higher than that of public transport, though the intrinsic utility of recycling was lower than that of public transport. 

\begin{figure}[hp]
	\parbox{.4\linewidth}{
	\includegraphics[trim=50 400 125 10, width=\linewidth]{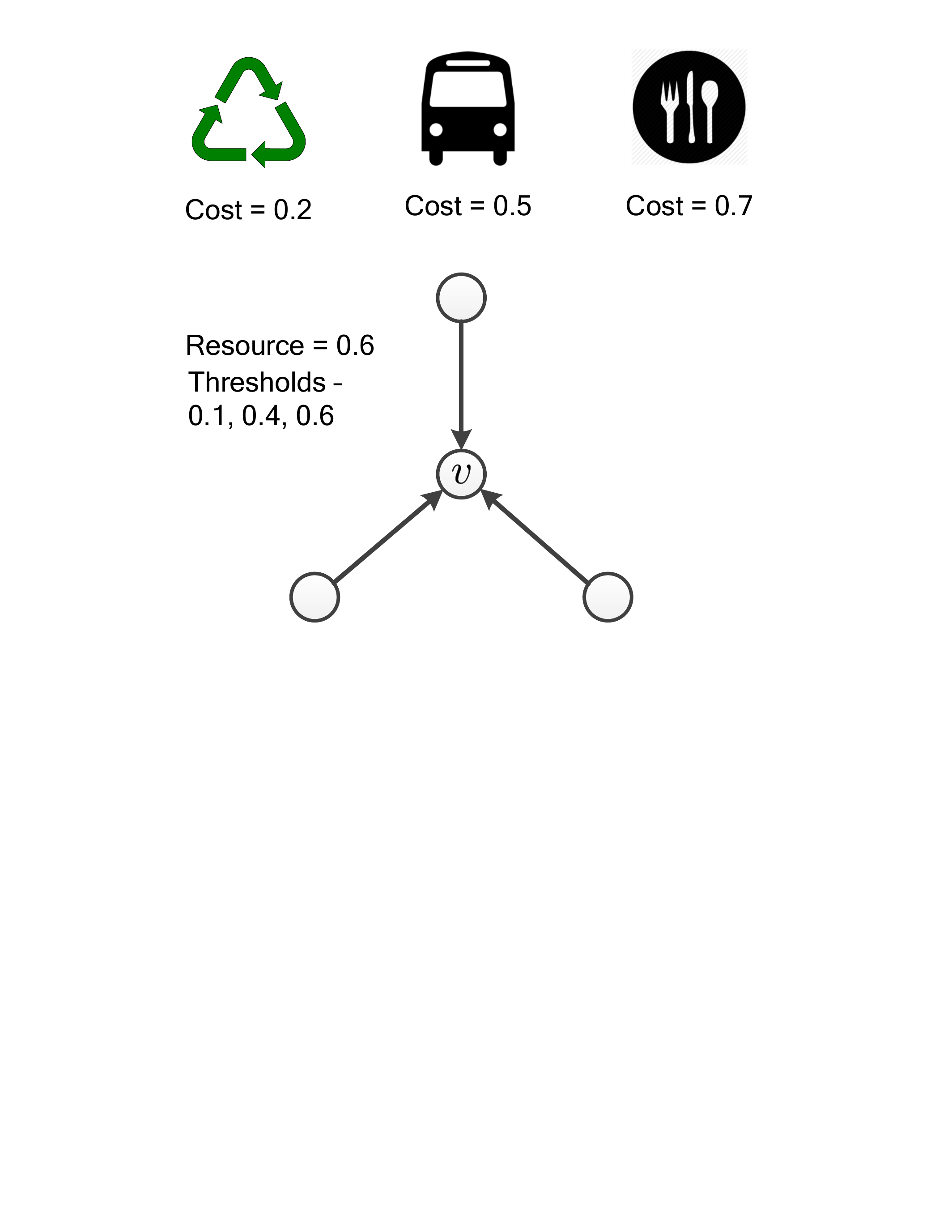}}%
	\hspace{.2\linewidth}%
	\parbox[][][t]{.4\linewidth}{%
	\subcaption{The three beahviors - (1) recycling, (2) using public transport, and (3) eating organic food with respective costs as well as the network is shown. The intrinsic utility of the behaviors are same as the cost. So $c_1=u_1=0.2, ~c_2=u_2=0.5, ~c_3=u_3=0.7$. Resource of the node $v$, $r(v)=0.6$, and the thresholds are - $\theta_1(v)=0.1,~\theta_2(v)=0.4,~\theta_3=0.6$.}}
	\parbox{.4\linewidth}{
	\includegraphics[trim=85 450 125 0, width=\linewidth]{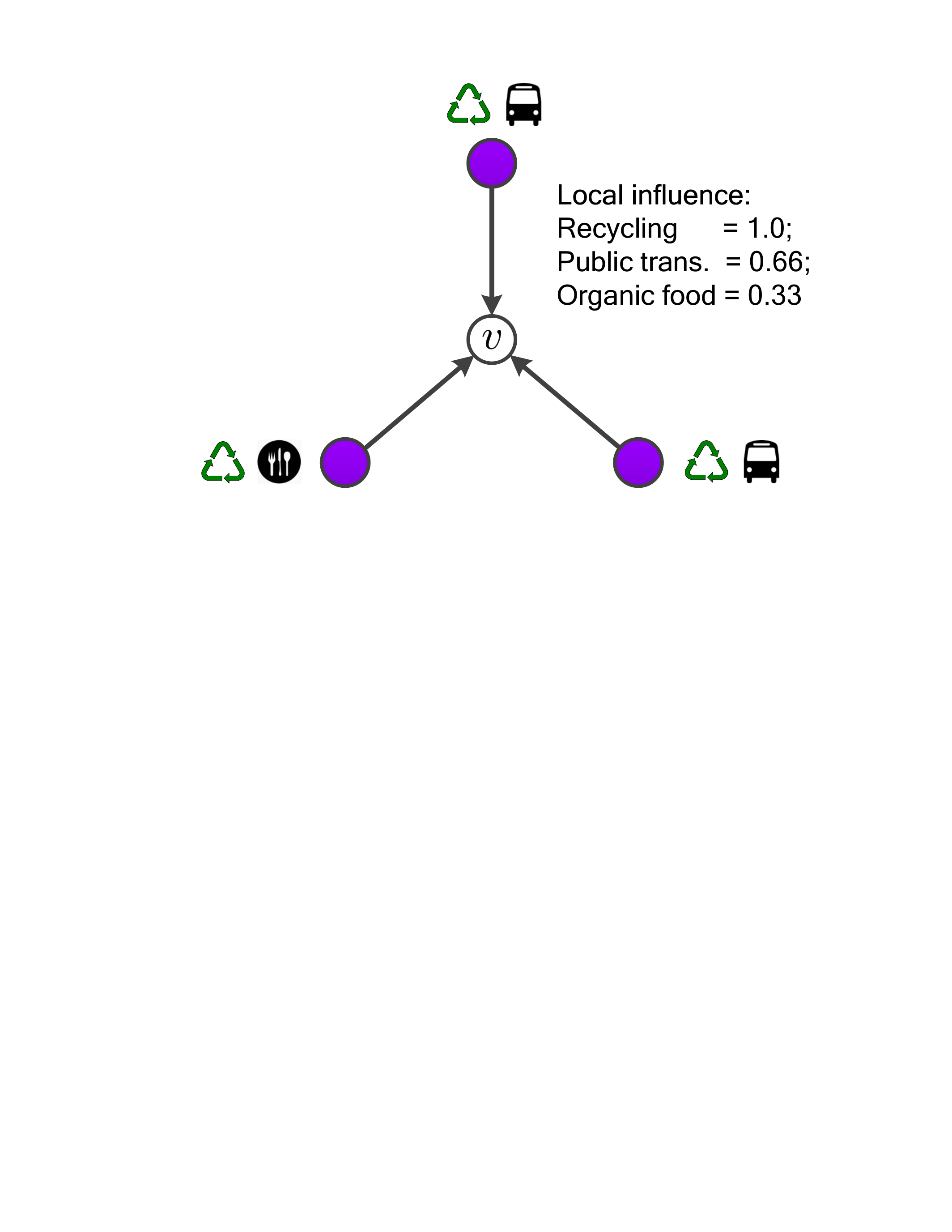}}%
	\hspace{.2\linewidth}%
	\parbox[][][t]{.4\linewidth}{%
	\subcaption{This is the network at time step $0$. All three of the neighbors of $v$ have adopted recycling, two of them have adopted public transport, and only one of them is eating organic food. $v$ has not adopted any behavior yet. The local influences for the three behaviors are as follows - $l_1(v)=1.0,~l_2(v)=0.66,~l_3(v)=0.33$.}\label{fig:state-before}}
	\parbox{.4\linewidth}{
	\includegraphics[trim=85 450 125 30, width=\linewidth]{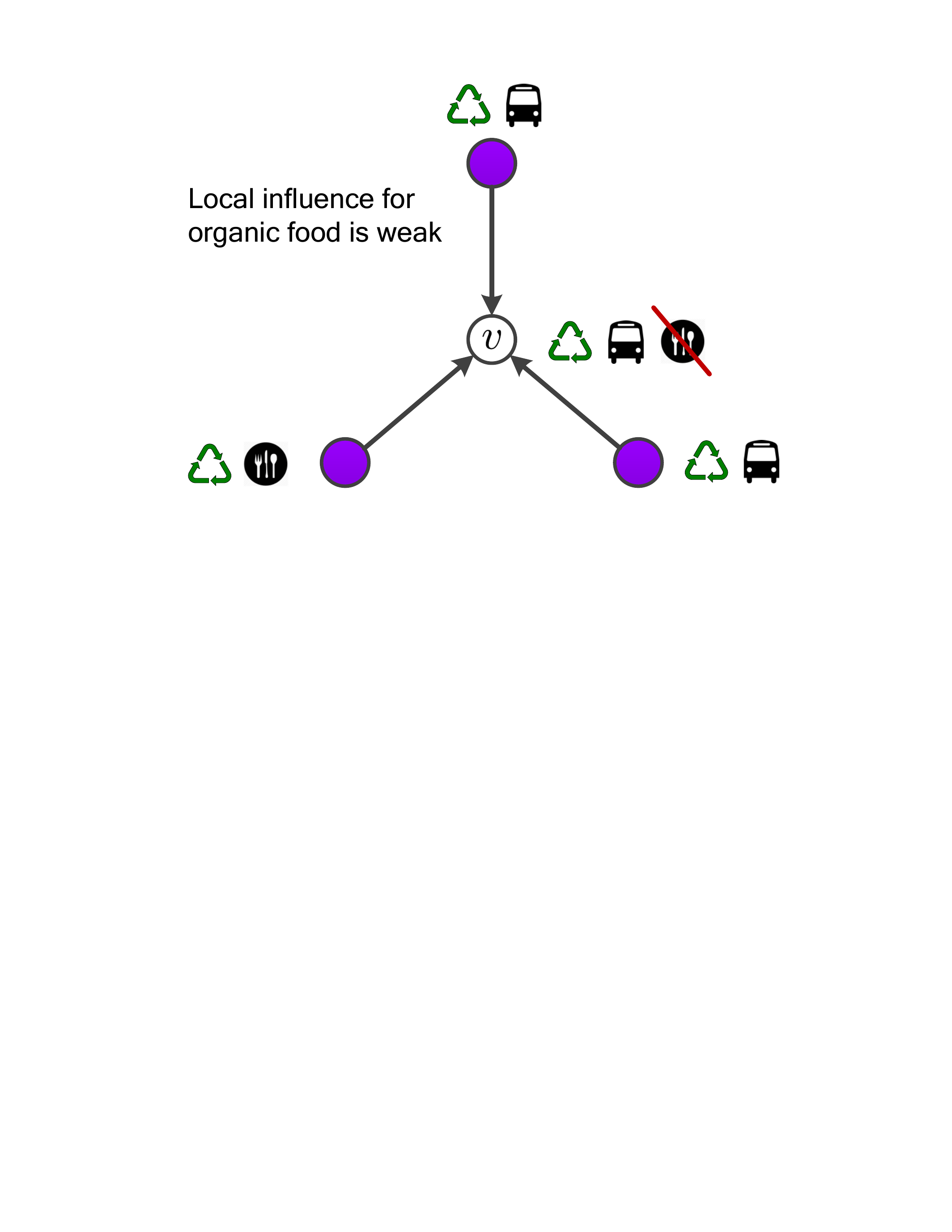}}%
	\hspace{.2\linewidth}%
	\parbox[][][t]{.4\linewidth}{%
	\subcaption{Local influence for organic food is less than the threshold, i.e $l_3(v)<\theta_3(v)$. So $v$ will not consider organic food for adoption.}}
	\parbox{.4\linewidth}{
	\includegraphics[trim=85 450 125 30, width=\linewidth]{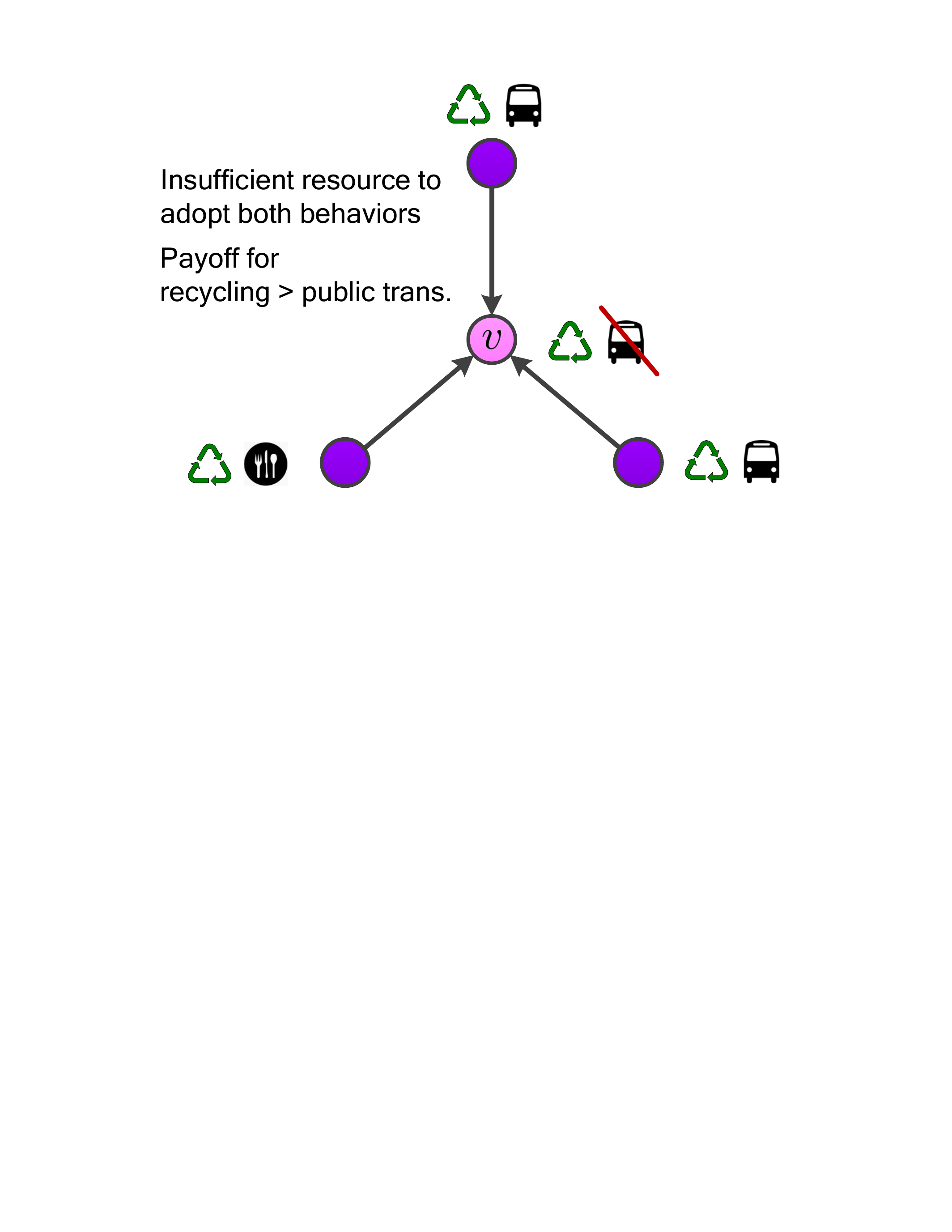}}%
	\hspace{.2\linewidth}%
	\parbox[][][t]{.4\linewidth}{%
	\subcaption{$c_1+c_2>r(v)$, so $v$'s resource is insufficient for adopting both recycling and public transport. Payoff for recycling, $p_1(v)=0.5\times 0.2+0.5 \times 1.0=0.6$, and public transport, $p_2(v)=0.5 \times 0.5 + 0.5 \times 0.66=0.58$ ($w=0.5$). For $v$ the payoff for recycling is higher than the payoff for using public transport though the intrinsic utility of public transport is higher than that of recycling. So $v$ will adopt recycling at the end of time step $1$. }}
	\caption{Multiple behavior adoption model.}\label{fig:model}
\end{figure}

For the purpose of deriving the analytical proofs, we shall assume that the adoption of behaviors is ``sticky'', that is, once a node adopts a behavior it never drops the behavior.  This simplification is also known as \textit{progressive} behavior adoption \parencite{Kempe03}.  Once a node adopts a ``sticky'' behavior its available resource to adopt other behaviors decreases. 

Having described the multiple behavior diffusion process, now we describe different measures of behavior diffusion.

\subsection{Measurement of Diffusion}\label{subsec:measure}
We measure the effectiveness of the diffusion process with three metrics: total participation, total adoption and resource utilization. The metrics are useful to understand the different seed selection heuristics discussed in the next section. Since the behavior adoption is a stochastic process, we compute the expected value of each metric through simulation.
\begin{description}
\item[Total Participation] This metric counts the number of individuals who have adopted at least one behavior (i.e. become active) during the process; we measure the expected number of adoptions via simulation. Exact computation of this metric is shown to be \#P-hard~\parencite{Chen10}). 
\item[Total Adoption] In contrast to total participation, we need to keep track of the total number adoptions of behaviors during the diffusion process.  We count the number of adoptions over all the behaviors during a diffusion process and via simulation determine the expected adoptions. Notice that since an individual can adopt more than one behavior, total participation can be less than total adoption. For the familiar single behavior adoption problem, these two measure will be identical.
\item[Resource Utilization] This metric captures the \textit{efficiency} of the network to adopt costly behaviors. 
Not all resources available in a social network may be used for behavior adoption. This is  because individuals have variable resources, and the may be unable to adopt the subset of behaviors that fully takes advantage of their desire to participate because of two reasons. First, they may have  more resources than needed to adopt a behavior. Second, if their friends (i.e. network neighbors) have limited resources, then the social signals that they receive will be about adopting low-cost resources, and hence a particular individual may never see costly behaviors in their social circle that they could potentially adopt. Let us assume that a node $v$ with resource $r(v)$ has adopted one or more behaviors. Let $s(v)$ be the amount of resource that $v$ has used to adopt those behaviors, where  $s(v) \le r(v)$. Therefore the individual has $r(v)-s(v)$ amount of his resource remaining unused. Thus, \textit{Resource utilization} is the expected value of the ratio $\sum_{v \in V}s(v) / \sum_{v \in V}r(v)$ i.e. the ratio of total utilized resource to the total amount of available resource of all the individuals in the social network. 
\end{description}

Having discussed the multiple behavior diffusion process under resource constraints and three measures of diffusion of behavior, we now turn to discuss the problem of identifying the top-$k$ influential nodes or seeds in a resource constrained network. There are two principal questions: how to efficiently identify influential nodes; for each node, assign the set of behaviors that this node ought to adopt. In the next section, we shall formally define the influential identification problem, show intractability and present an approximation algorithm.  In Section \ref{sec:heuristics}, we shall introduce several heuristics to identify influential nodes and well as show experimental results. In Section \ref{sec:dist}, we address the question of distributing behaviors onto the influential nodes.

\section{The Influential Identification Problem} \label{sec:seed-selection}
In this section we introduce the first algorithmic problem that we want to address in the context of multiple behavior diffusion in a resource constrained social network. This problem is called \textit{seed selection} problem. In the next section we will formally define this problem and analyze its complexity. Then we will describe and analyze different strategies for solving the problem. Then we will provide experimental evaluation of those strategies on synthetic as well as real world social networks.

\subsection{Problem Definition}
There are two key problems: we need to identify the set of early adopters or seed nodes and we need to determine which behaviors ought to be adopted by each seed node. We assume that the number of initial adopters is small in comparison to the size of the network. This is reasonable as it corresponds to an advertiser with a finite budget to persuade the seeds to adopt. Here we identify two subproblems which are related to seed identification. To simplify things, in this section we will assume that the behaviors are uniformly distributed over the seed set.

\begin{description}
\item[\textbf{P1: Resource Utilization Maximization}] Given a fixed seed budget $b$ and a fixed distribution of behaviors in the seed set, we want to select $b$ nodes in the network such that the resource utilization metric is maximized.
\item[\textbf{P2: Total Participation (or Adoption) Maximization}]  Given a fixed seed budget $b$ and a fixed distribution of behaviors in the seed set, we are interested in finding $b$ nodes in the network that maximize the total participation (or total adoption) in the network.
\end{description}

It can be easily shown that the optimization problems P1 and P2 are NP-complete. We show that influence maximization problem for LT model, which is proven to be an NP-complete problem \cite{Kempe03}, is a special case of P1. Let the number of behaviors $k=1$ and the cost of adoption of that behavior is also $1$. Each node $v$ is allocated resource $r(v)=1$. For these values of the parameters our multiple behavior diffusion model reduces to the LT model of influence propagation and resource utilization can be calculated as the ratio of the spread and total number of nodes in the network. So maximizing the resource utilization translates into maximizing the spread. Same transformation applies to problem P2 also since total participation (and total adoption) is identical to the spread in the one behavior case. Next, we propose a number of heuristics to solve the problem.

\subsection{Approximation Algorithm for the Sticky Model} \label{sub-sec:approx-algo}
In this section we will provide an approximation algorithm for a variant of the problem P2 (total participation maximization) when behavior adoption is sticky --- i.e. once an individual adopts a behavior he never abandons that behavior. Note that both problems---P1, P2 are still NP-hard in the Sticky model. We will construct our algorithm by first showing that the total participation function is submodular. Then following a technique by \cite{Nemhauser78} we will obtain an approximation algorithm for the problem.

\subsubsection{Submodularity of Total Participation}
In this section we will show that the total participation function for the sticky model is submodular. Throughout this section we will assume that the $k$ behaviors are indexed in the ascending order of their cost. For a given set $S$ of seeds for the $k$ behaviors, let $\sigma(S)$ denote the \emph{total participation} i.e. the expected number of active nodes, regardless of behavior adopted, at the end of the process. We will show that $\sigma(S)$ is a submodular function. 

\begin{theorem}
    For an arbitrary instance of the Sticky Multiple Behavior Diffusion model the total participation function $\sigma(.)$ is submodular.
\end{theorem}

\begin{proof}
    The proof consists of three steps---first we define an equivalent alternative process of the Sticky Multiple Behavior Diffusion process, then show that the alternative process and the sticky multiple behavior diffusion process are both equivalent in the sense of the distribution and finally prove the submodularity for the equivalent process.
    
    We introduce an alternative live edge process to model behavior diffusion since the threshold model for behavior adoption is known to be not sub-modular.~\cite{Kempe03}. For the alternative model we would like to distinguish between the behaviors a node can adopt under \textit{some} circumstances and those that it cannot under \textit{any} circumstances. If the cost of adoption of a behavior is greater than the available resource of a node then clearly the node canot adopt that behavior. Since the behaviors are indexed in the ascending order of their cost, this distinction can be made by a single indicator variable. For each node $v$, let us define $\kappa(v)$ as the largest index $j$ such that $c_j\le r(v)$. If $v$ does not have enough resource to adopt any of the behaviors then $\kappa(v)$ is defined to be $0$. 
    
    Prior to any diffusion each node $v$ selects at most $\kappa(v)$ edges---thus once for each behavior that it can potentially adopt---by repeating with replacement the following random edge selection process $\kappa(v)$ times. Each node $v$ selects the edge $e_{v,w}$ with probability $b_{v,w}$ and no edge with probability $1-\sum_{w \in N(v)}b_{v,w}$, where $b_{v,w}$ denotes the influence weight exerted by the neighbor $w$ on $v$ and $\sum_{w \in N(v)}b_{v,w} \leq 1$. So $v$ selects at most one edge for each of the $\kappa(v)$ behaviors. 
    
    If an edge is selected then that edge is designated as the \emph{live} edge for the behavior $i$.  All the other edges are considered \emph{blocked} for that behavior. In this model we start with one set of seeds for each of the $k$ behaviors, $k$ sets in total. In time step $t$, a node $v$ considers a behavior $i$ for adoption such that: the behavior $i$ is not already adopted by $v$ and that a neighbor of $v$, connected to $v$ by the live edge for the behavior $i$, has adopted behavior $i$ by time $t-1$. The node then adopts a subset of all such  considered behaviors to maximizes its total payoff subject to the constraint that the combined cost of the behaviors is less than the available resource. Once a node adopts a behavior it becomes \textit{active} with respect that behavior and the behavior sticks---when a node adopts a behavior it never gets rid of it.  
    
    Next we prove a lemma that the alternative live process is stochastically equivalent to the Sticky Multiple Behavior Diffusion model described earlier.
    
    \begin{lemma}\label{lemma:alternative-proc}
            For a given seed set $S$, the following two distributions over the set of nodes are the same:
            \begin{enumerate}
                \item The distribution over active sets obtained by running the sticky multiple behavior diffusion model to completion starting with $S$. 
                \item The distribution over sets of active nodes reachable from $S$ via live edges under the random selection of edge model described above.
            \end{enumerate}
    \end{lemma}
    
    \begin{proof}
		First we prove for the simpler case when $k=1$, i.e. there is only one behavior that diffuses in the network and then generalize to the case of arbitrary number of behaviors.  Notice that in this case each node will have at most one live edge. This case is  similar to the Linear Threshold (LT) model described in \cite{Kempe03} with the key difference that in our case each node is resource constrained. Let $G=(V,E)$ be the given graph. We only need to consider the nodes $v$ with $r(v)\geq c_1$. We delete all nodes with $r(v)<c_1$ and  edges associated with these nodes from $G$ producing a graph $G'$. Thus in $G'$, the behavior adoption process  reduces to the single behavior LT model discussed in \cite{Kempe03}. 
		
		Now we repeat that proof from~\citeN{Kempe03} since it will serve as the basis for the generalized case of $k >1$.  We argue by induction over the time step $t$. Let $S^{(t)}$ be the set of nodes who have adopted behavior $1$ at the end of time step $t$ for the sticky  behavior diffusion model with $k=1$. We need to know the probability that a node $v$ with $r(v)\ge c_1$ that have not yet adopted behavior $1$ at the end of time step $t$ will adopt the behavior in the next time step $t+1$. This probability is the same as the probability that the nodes in $S^{(t)}\setminus S^{(t-1)}$ will push the influence weight of $v$ over its threshold at time $t$, given that the threshold was not already crossed at time $t-1$. This probability is given by:
		
		\begin{equation*}
		\frac{\sum_{w\in S^{(t)}\setminus S^{(t-1)}} b_{v,w}}{1-\sum_{w\in S^{(t-1)}}b_{v,w}}. 
		\end{equation*}
		
		For the alternative live edge process each node $v$ with $r(v)\ge c_1$ selects at most one live edge randomly at the beginning of the process. Notice that we don't have to select the live edges over time, they can all be simply selected in the beginning.  Under this model we need to compute the probability that a node $v$ with $r(v) \ge c_1$ that has not adopted behavior $1$ at the end of time step $t$ will adopt it in the next time step. This probability is precisely same as the probability that the live edge of $v$ comes from one of the nodes in $S^{(t)}\setminus S^{(t-1)}$, given that it did not came from $S^{(t-1)}$. This probability is also given by:
		
		\begin{equation*}
		\frac{\sum_{w\in S^{(t)}\setminus S^{(t-1)}}b_{v,w}}{1-\sum_{w\in S^{(t-1)}}b_{v,w}}.
		\end{equation*}
		
		So by induction we find that the two processes define the same distribution over the active sets. Next we provide the proof for when the number of behaviors $k>1$. Notice that~\cite{Kempe03} only proved the case of $k=1$.
	    
		We prove the claim by induction on the time step $t$. Clearly the claim is true for $t=0$. We define $S^{(t)}_{i}$ as the set of active nodes with behavior $i$ at the end of time step $t$ of the sticky
		multiple behavior diffusion model. Let $S^{(t)}:=\cup_{i=1}^{k}S^{(t)}_{i}$.  That is, $S^{(t)}$ is the set of nodes that have adopted at least one behavior.  Notice that $S^0=S$. Let $v$ be a node that is not active at the end of time step $t$ and $\kappa(v)=\kappa \ne0$, where $\kappa(v)$ is the number of distinct behaviors $v$ can adopt one at a time. Then the probability that $v$ will become active at the end of time step $t+1$ is equal to the probability that the nodes in $S^{(t)}\setminus S^{(t-1)}$ will push the influence weight of at least one of the first $\kappa$ behaviors over its corresponding threshold value, given that none of those thresholds were already crossed. This probability is 
		 
		 \[
		 1-\prod_{i=1}^{\kappa}\left(1-\frac{\sum_{w\in S^{(t)}_{i}\setminus S^{(t-1)}_{i}}b_{v,w}}{1-\sum_{w\in S^{(t-1)}_{i}}b_{v,w}}\right)
		 \]

		On the other hand we run the live edge reachability process as described above and denote by $S^{(t)}_{i}$ the set of all nodes with behavior $i$ at the end of time step $t$. Let $S^{(t)}:=\cup_{i=1}^{k}S^{(t)}_{i}$. If node $v$ is not active at the end of time step $t$ with $\kappa(v)=\kappa \ne 0$, then the probability that it will be active at the end of time step $t+1$ is equal to the probability that at least one of its $\kappa$ live edges comes from the nodes of $S^{(t)}\setminus S^{(t-1)}$ (with the corresponding behavior), given that none of those live edges came from $S^{(t-1)}$. This probability is also given by - 
		 
		 \[
		 1-\prod_{i=1}^{\kappa}\left(1-\frac{\sum_{w\in S^{(t)}_{i}\setminus S^{(t-1)}_{i}}b_{v,w}}{1-\sum_{w\in S^{(t-1)}_{i}}b_{v,w}}\right)
		 \]
		 
		By induction over the time step of the process we see that the distribution over the active sets at the end of the sticky multiple behavior diffusion process is same as the distribution produced by the alternative live edge process.
	\end{proof}
	    
	Let us define $\sigma'(S)$ as the expected number of active nodes at the completion of the alternative random process.  By the previous lemma \ref{lemma:alternative-proc}, we can show that $\sigma'(S)=\sigma(S)$, for all seed sets $S$ under same distribution of behaviors on the seed set, and where $\sigma(S)$ is the total participation (expected number of active nodes, active for any behavior) for the sticky multiple behavior adoption process. 
	
	Now we show that $\sigma'(.)$, the total participation function is submodular. Let $X$ be a particular choice of live/blocked edges for all nodes. Let $\sigma'_X(S)$ denote the cardinality of the set of active nodes at the completion of the alternative process. Let $R(v,X)$ denote the set of nodes reachable from seed node $v$ that satisfy the following condition --- any node in $R(v,X)$ has a behavior seeded by $v$ and the node is connected to $v$ by a live edge path for that behavior under the choice $X$.  Therefore $\sigma_X'(S)=|\cup_{v\in S}R(v,X)|$.
	 
	First we will show that for a fixed choice $X$, $\sigma'_X(.)$ is submodular. Let $S$ and $T$ be two sets of nodes such that $S\subseteq T$ and $v$ is any node. Let us consider $\sigma'_X(S\cup\{v\}) - \sigma'_X(S)$. This is the number of nodes that are in $R(v,X)$ but not in $\cup_{u\in S}R(u,X)$. This number is at least as large as the number of nodes in $R(v,X)$ but not in the bigger union $\cup_{u\in T}R(u,X)$. Therefore it follows that $\sigma'_X(S\cup\{v\})-\sigma'_X(S)\ge \sigma'_X(T\cup \{v\})-\sigma'_X(T)$.
	 
	Finally we have 
	 
	 \[
	 \sigma'(S)=\sum_{\mathrm{outcomes\, }X} \mathrm{Pr}[X].\sigma'_X(S) 
	 \]
	 
	Since a non-negative linear combination of submodular functions is also submodular, $\sigma'(.)$ is submodular. This completes our proof.
\end{proof}

\subsubsection{The Approximation Algorithm}
We are interested in obtaining an approximation guarantee for the total participation maximization problem under the Sticky multiple behavior diffusion model. For this type of optimization problems involving submodular functions there is a greedy algorithm that approximates the optimum within a factor or $(1-1/e-\epsilon)$, where $e$ is the base of natural logarithm and $\epsilon$ is any positive real number (\cite{Nemhauser78}, \cite{Kempe03}). So the approximation algorithm gives a performance guarantee of at least $63\%$ of the optimum. We modify the basic greedy algorithm to adapt it to the multiple behavior case (Algorithms \ref{algo:approx-sticky}, \ref{algo:core}). 

\begin{algorithm}[h]
\KwIn{$G:=(V,E)$, the social network; $\mathbf{b}$, a vector of size $k$ containing number of required seeds for each of the behaviors.} 
\KwOut{$\mathbf{S}$, a vector of size $k$ containing seed sets of required size for all the behaviors.} 
Let $V':=V$ and $\mathbf{S:=}\boldsymbol{\phi}$\; 
\Repeat{$\mathbf{b}=\mathbf{0}$}{     
	\For{each behavior $i$}{        
		Let $(u_i , s_i ):=$Core-Greedy($i, \mathbf{b}[i], \mathbf{S}, V'$) \; \nllabel{line:inc-core-hill-climbing}
    }     
	Let $i_{max}:=\argmax_{i\in \{1,\ldots,k\}}s_i $ \; 
	Let $v:=u_{i_{max}}$ \;   
	Set $\mathbf{S}[i_{max}] := \mathbf{S}[i_{max}] \cup \{ v \}$ and $\mathbf{b}[i_{max}] := \mathbf{b}[i_{max}] - 1$\;  
	Set $V':=V'-v$ \;        
	\If{$r(v) \le c_{i_{max}}$} {
       Set $r(v):=c_{i_{max}}$\;
    } 
} 
\caption{Approximation algorithm for the sticky multiple behavior diffusion model} 
\label{algo:approx-sticky}
\end{algorithm}

In Algorithm~\ref{algo:approx-sticky}, in line 4, we obtain the node $u_i$ associated with the maximum spread $s_i$ for behavior $i$ via the \texttt{Core-Greedy} algorithm. Then, in lines 6--7, we identify the behavior $i_{\max}$ with the maximum spread and the corresponding node $u_{i_{\max}}$ that was the seed for the behavior. This node $v$ is then added to the set of seeds for that behavior ($i_{\max}$, line 8) and then removed from the set of nodes to be considered as seeds in the next round (line 9). In line 11, if the resource for the seed is less than the cost of the behavior, we "top-up" the resource of the seed so that it can adopt the behavior.  Notice that to avoid "topping-up" we can simply ignore nodes $v$ that have resources less than the cost of the behavior to be adopted $r(v) < c_i$ when examining nodes in the  \texttt{Core-Greedy} algorithm. 

\begin{algorithm}[h]
\KwIn{$i$ the behavior; $\mathbf{b}[i]$, the number of seeds required for the $i$th behavior; $\mathbf{S}$, the set of already selected seeds for all the behaviors; $V'$ the remaining population of nodes from where one chooses new seeds.} 
\KwOut{$(u,s)$ if $\mathbf{b}[i]$ is not zero then a tuple consisting of $u$, the best choice of seed from the population $V'$ for the $i$th behavior, given the already selected seedset $\mathbf{S}$, and $s$ its corresponding spread value (total participation).}
   
\If {$\mathbf{b}[i] = 0$} {
	Return (`nobody',$0$)
}
Let $\mathbf{s}$ be a vector indexed by the set $V'$, and $\mathbf{s}=\mathbf{0}$ \; 
\For {$v \in V'$} {
	Let $\mathbf{S'} := \mathbf{S}$ \;
	Set $\mathbf{S'}[i] := \mathbf{S'}[i] \cup \{v\}$ \;
	Set $\mathbf{s}[v] :=$ Estimate-Spread($\mathbf{S'}$) \; \nllabel{line:inc-estimate-spread}
}
Select $u:=\argmax_{v}\{\mathbf{s}[v]|v \in V'\}$ \;    
Return $(u,\mathbf{s}[u])$;
\caption{Core-Greedy algorithm used in the approximation algorithm for the sticky model} 
\label{algo:core} 
\end{algorithm}

In Algorithm~\ref{algo:core}, we present the \texttt{Core-Greedy} algorithm which selects seeds given a behavior, the set of nodes from which to choose the ``best" node, the set of nodes already chosen to be seed nodes.  In line 8, we estimate the spread of the behavior through a stochastic simulation: given a specific node to be selected as the seed node, we compute the expected spread via a simulation; in each run, we assign each node with a threshold picked from $U(0,1)$ and then let the behavior spread by assessing for each node that hasn't yet adopted the behavior whether it will adopt the behavior. Finally in lines 10--11 we select the node with the highest expected spread and return this tuple. 

In this section we formally defined the seed selection problem and showed that the problem is NP-hard. Then in Section~\ref{sub-sec:approx-algo} we showed the total participation function was sub-modular leading to a greedy seed selection algorithm.

Although we achieve an approximation guarantee for the seed selection problem, the time complexity of the approximation strategy is $O(n^2kb)$, where $n$ is the number of vertices in the graph, $k$ is the number of behaviors, and $b$ is the number of seeds required. Moreover the constant is  large (\citeN{Kempe03} use 10,000 simulations to estimate the expected spread) since we need to simulate the diffusion process multiple times to estimate the spread throughout the algorithm. This makes the approximation algorithm impractical for networks of large size and motivates us to explore cheap heuristics that perform reasonably well for practical instances. Different efficient heuristics for the seed selection problem is the topic of the next section.

\section{Seed Selection: Heuristics and Experiments} 
\label{sec:heuristics}
In this section we develop heuristics to identify seed nodes for behavior diffusion in a resource constrained social network and then show experimental results that help us analyze the impact of each heuristic.

\subsection{Summary of Variation \& Notation }
\label{subsec:notation}
There are a number of variations possible---in terms of resources and behaviors---for nodes selected as seeds. The first issue is whether we \textit{top up} the resource of a selected seed. \textit{Topping up} a seed  provides it with additional resources to adopt the behavior in case the behavior was more costly than available resources. This corresponds to real life events like providing early adopters with with free items, gift coupons or other services like free access to recycling facilities etc. Depending on whether we allow seeds to be topped up we have two variations of the seed selection algorithm---\textit{Topped Up} (suffix \textbf{T} is added to the name of the algorithm) and \textit{No Top Up} (suffix \textbf{NT} is added). In the \textbf{NT} version only the nodes with sufficient resource for adopting a behavior are considered as candidates for seed selection. On the other hand in the \textbf{T} version all the nodes
are considered as possible candidates.

Another variation is possible depending on whether a node can be selected as a seed for more than one behavior or strictly one behavior. In the first case a seed may be assigned more than one behaviors and we suffix \textbf{M} (for \textit{multiple}) to the seed selection algorithm. In the second case a seed is assigned exactly one behavior and we use the suffix \textbf{S} (for \textit{single}). It is easy to see that \textbf{S} version can never find a solution that is better than the \textbf{M} version for the same type of top up regime.

Combining these two types of variations we can have four different variants of each seed selection algorithm - \textbf{S-T}, \textbf{S-NT}, \textbf{M-T} and \textbf{M-NT}. In this paper most of the results are for the \textbf{S-T} variant. However in Appendix \ref{sec:t-nt-comp} we present some results comparing these different variations of the seed selection algorithm and discuss a few consequences.

\subsection{Seed Selection Heuristics}
\label{subsec:heuristics}
We discuss heuristics based on node degree, influence weight  and expected immediate adoption for the behavior diffusion model under resource constraint. We start by developing heuristics which are based on node degree. 
\subsubsection{Node Degree}
The social capital of an individual increases the number of acquaintances. While the nature of the connections and the specific structure of the network in which an individual is embedded play a role in determining the influence of an individual, an individual's node degree is a good first degree approximation to an individual's ``influence'' on his acquaintances. We first discuss the heuristic and present some useful variants.

\begin{description}
\item[\textbf{Na\"ive}] We rank the nodes according to their degree, pick top $k$ nodes, and assign them different behaviors. This is a na\"ive extension of the high degree heuristic for the LT model \cite{Kempe03}. We test three variants of this heuristic. The first variant is \textit{na\"ive degree with random tie breaking and no top up} (see Algorithm~\ref{algo:naive-rand-no-topup}). Here \textit{no top up} means that a seed node is never assigned a behavior that it cannot adopt with its own resource, i.e. a seed node is assigned a behavior only if it has sufficient resource to adopt it. \textit{Random tie breaking} means that we assign a high degree node a behavior that is chosen uniformly at random from the set of behaviors that the node can adopt with its available resource. So each seed node is assigned at most one behavior (if a seed node does not have enough resource to adopt any behavior, no behavior is assigned to it).
In Algorithm~\ref{algo:naive-rand-no-topup} the loop (lines 2-9) continues until we have no more nodes to identify as seeds. In the second variant \textit{na\"ive degree with random tie breaking and top up} each seed node is always assigned one randomly chosen behavior irrespective of its resource level. If its resource is not sufficient for adoption of the behavior we top up its resource so that it can bear the cost of adoption of the assigned behavior. In the third variant \textit{na\"ive degree with knapsack tie breaking}, each seed node is assigned all the behaviors that will maximize its utility subject to its resource constraint---to determine which behaviors to assign we solve a knapsack problem. Notice that degree based heuristics are optimistic---it is possible that neighbors of a seed do not have resources to adopt the behavior of the seed, thus preventing diffusion of behavior.

\begin{algorithm}[t]
\KwIn{$G:=(V,E)$ - the social network, $\mathbf{b}$ - a vector of size $k$ containing number of required seeds for each of the behaviors}
\KwOut{$\mathbf{S}$ - a vector of size $k$ containing seed sets for each of the $k$ behaviors}
Let $V':=V$ and $\mathbf{S}:=\boldsymbol{\phi}$\;
\Repeat{$\mathbf{b}=\mathbf{0}$}{
    Select $v:=\argmax_u\{|N(u)| : u\in V'\}$\;
    $V':= V' \setminus \{v\}$ \;
    Select $j$ uniformly at random from the set of behaviors $i$ that still need seeds to be assigned: $\{i:\mathbf{b}[i] \ne 0 \}$ \;
    \If{$r(v) \ge c_j$} {
       Set $\mathbf{S}[j] := \mathbf{S}[j] \cup \{v\}$ and $\mathbf{b}[j] := \mathbf{b}[j] - 1$\;
       Designate $v$ as an early adopter for behavior $j$\;
    }
   }
\caption{Na\"ive Degree Based with Random Tie breaking and No Top Up}
\label{algo:naive-rand-no-topup}
\end{algorithm}

\item[\textbf{Neighbors With Sufficient Resource}]  This heuristic takes into account both the degree and available resource of the neighbors when selecting the seed nodes. For each behavior $i$ we calculate $d_i (v)$ --- the number of neighbors of a node $v$ with sufficient resource for adoption of behavior $i$ with cost $c_i$ (i.e. the number of neighbors $u$ of seed node $v$ with $r(u) \ge c_i $). Clearly $d_i (v)$ is a better indicator of the suitability of selecting $v$ as a seed for the $i$th behavior than just the node degree. In the \textit{degree and resource ranked} heuristic (see Algorithm~\ref{algo:degree-resource}) we compute $d_i (v)$ for all the nodes (lines 2-10), rank them according to the value of this metric and select the required number of seeds for the $i$th behavior from the top of the ranking. If a node is selected as a candidate seed for more than one behaviors, we break the tie randomly and top up its resource so that it can adopt the randomly assigned behavior (lines 19-22). Then we add one more candidate for the behaviors that were not assigned. We repeat the process until the required number of seeds are selected for all the behaviors. For example if we need to identify $3$ seeds for behavior $1$, $4$ for behavior $2$, and $3$ for behavior $3$. Assume that the second ranked node on the list for behavior $1$ appears in all the other lists. Suppose when we break the tie the node is assigned behavior $3$. Then we pick the one additional node for each behavior $1$ and $2$ with the next highest value of $d_i(v)$.

\end{description}
One weakness of the degree based and degree-resource based heuristics is that they provide no estimate of the effectiveness of the seed in terms of adoptions. We address this issue next. 

\begin{algorithm}[t]
\KwIn{$G:=(V,E)$ - the social network, $\mathbf{b}$ - a vector of size $k$ containing number of required seeds for each of the behaviors}
\KwOut{$\mathbf{S}$ - a vector of size $k$ containing seed sets of required size for all the behaviors}
Let $d_i (v):=0$ for all $v\in V$ and $i \in \{1,\ldots,k\}$\;
\For{each $v \in V$}{
    \For{each behavior $i$}{
       \For{each neighbor $u$ of $v$}{
          \If{$r(u)\ge c_i $}{
             $d_i (v):=d_i (v)+1$\;
          }
       }
    } 
}
Let $V':=V$ and $\mathbf{S:=}\boldsymbol{\phi}$\;
\Repeat{$\mathbf{b}=\mathbf{0}$}{
    \For{each behavior $i$}{
       Let $T_i $ be the set of top $\mathbf{b}[i]$ nodes from $V'$ in the decreasing sorted order of $d_i (v)$\; 
    }
    Let $T:=\cup_{i=1}^{k}T_i $\;
    Set $V':=V'\setminus T$ \;
    \For{each node $v$ in $T$}{
       Select $j$ uniformly at random from the set of behaviors \{$i | v\in T_i $ \} \;
       \If{$r(v) \le c_j$} {
          Set $r(v):=c_j$\;
       }
       Set $\mathbf{S}[j] := \mathbf{S}[j] \cup \{v\}$ and $\mathbf{b}[j] := \mathbf{b}[j] - 1$\;
       Designate $v$ as an early adopter for behavior $j$\;
    }
   }
\caption{Degree and Resource Ranked Heuristic}
\label{algo:degree-resource}
\end{algorithm}

\subsubsection{Influence Weight Based Heuristics}
We compute an influence weight measure to estimate the influence of a potential seed set on its neighbors.  We can compute the influence weight measure for a set of seeds by summing over the influence weight of individual seeds.  Let $u$ be a neighbor of $v$. $v$ exerts a social influence of weight $\scriptstyle 1/|N(u)|$ on $u$. 
The Influence Weight exerted by $v$ on its neighbors is $\scriptstyle \sum_{u \in N(v)} \frac{1}{|N(u)|}$.  Notice that while social capital of a node $v$ is given by $N(v)$ the number of neighbors of $v$, its net social influence depends on the number of neighbors of its immediate neighbors. That is social influence of a node $v$ is $\scriptstyle \propto  \sum_{u \in N(v)} \frac{1}{|N(u)|}$.

We will restrict the summation over those neighbors $u$ that have enough resource to adopt behavior $i$. Hence we call this metric \textit{Constrained Social Influence Weight (CIW)} of $v$ for the behavior $i$ and denote it by $e_i (v)$. The justification of this heuristic comes from \cite{Watts07}, where it is argued that large cascades are driven by a critical mass of easily influenced individuals. However, we note that \citeN{Watts07} did not consider resource bounded individuals in their framework. Intuitively, higher $e_i(v)$ includes the possibility that $v$ is connected to individuals who can be easily influenced by $v$ to adopt behavior $i$ which can potentially lead to a large cascade of behavior $i$.  Next we describe two heuristics based on the CIW. 
 
 \begin{description}
\item[\textbf{Rank Based}] We rank all the nodes (see Algorithm~\ref{algo:ond-step-ranked}) based on the value of $e_i (v)$ (lines 2--10) and choose the required number of seeds for behavior $i$ starting from the highest ranked nodes. We perform the same evaluation for all behaviors. If a node is selected as a candidate seed for more than one behaviors, then we pick one of the behaviors at random and assign it to the node (line 19).  And similar to Algorithm~\ref{algo:degree-resource}, for all behaviors not selected we pick from the remaining nodes (i.e. $v \in V'$) with the next highest value of $e_i(v)$ and add it to the corresponding list (i.e. $T_i$ for a behavior not selected). If the node does not have sufficient resource to adopt that behavior then its resource is topped up. The process continues until the required number of seeds are allocated to all the behaviors.

\begin{algorithm}[t]
\KwIn{$G:=(V,E)$ - the social network, $\mathbf{b}$ - a vector of size $k$ containing number of required seeds for each of the behaviors}
\KwOut{$\mathbf{S}$ - a vector of size $k$ containing seed sets of required size for all the behaviors}
Let $e_i (v):=1$ for all $v\in V$ and $i \in \{1,\ldots,k\}$\;
\For{each $v \in V$}{
    \For{each behavior $i$}{
       \For{each neighbor $u$ of $v$}{
          \If{$r(u)\ge c_i $}{
             $e_i (v):=e_i (v)+\frac{1}{|N(u)|}$\;
          }
       }
    } 
}
Let $V':=V$ and $\mathbf{S:=}\boldsymbol{\phi}$\;
\Repeat{$\mathbf{b}=\mathbf{0}$}{
    \For{each behavior $i$}{
       Let $T_i $ be the set of top $\mathbf{b}[i]$ nodes from $V'$ in the decreasing sorted order of $e_i (v)$\; 
    }
    Let $T:=\cup_{i=1}^{k}T_i $\;
    Set $V':=V'\setminus T$ \;
    \For{each node $v$ in $T$}{
       Select $j$ uniformly at random from the set of behaviors \{$i | v\in T_i $ \} \;
       \If{$r(v) \le c_j$} {
          Set $r(v):=c_j$\;
       }
       Set $\mathbf{S}[j] := \mathbf{S}[j] \cup \{v\}$ and $\mathbf{b}[j] := \mathbf{b}[j] - 1$\;
       Designate $v$ as an early adopter for behavior $j$\;
    }
   }
\caption{CIW Rank Based}
\label{algo:ond-step-ranked}
\end{algorithm}

\item[\textbf{Max Margin}] The hill climbing heuristic builds the seed set incrementally, each time selecting a new seed node that maximizes the marginal increase of influence weight until the required number of seeds is obtained. So in this case while computing the influence weight of a potential seed node for a behavior we exclude those neighbors of the node who are already selected as seeds for that behavior. As with the previous heuristic, if the node does not have sufficient resource then it is topped up so that it can adopt the assigned behavior.

\begin{algorithm}[t]
\KwIn{$G:=(V,E)$ - the social network, $\mathbf{b}$ - a vector of size $k$ containing number of required seeds for each of the behaviors}
\KwOut{$\mathbf{S}$ - a vector of size $k$ containing seed sets of required size for all the behaviors}
Let $V':=V$ and $\mathbf{S:=}\boldsymbol{\phi}$\;
\Repeat{$\mathbf{b}=\mathbf{0}$}{
    \For{each behavior $i$}{
       Let $T_i :=$Core-Hill-Climbing($i, \mathbf{b}[i], \mathbf{S}[i], V'$) \;
    }
    Let $T:=\cup_{i=1}^{k}T_i $\;
    Set $V':=V'\setminus T$ \;
    \For{each node $v$ in $T$}{
       Select $j$ uniformly at random from the set of behaviors $i$ with $v\in T_i $ \;
       \If{$r(v) \le c_j$} {
          Set $r(v):=c_j$\;
       }
       Set $\mathbf{S}[j] := \mathbf{S}[j] \cup \{v\}$ and $\mathbf{b}[j] := \mathbf{b}[j] - 1$\;
       Designate $v$ as an early adopter for behavior $j$\;
    }
   }
\caption{CIW based Max Margin Heuristic}
\label{algo:one-step-hill-climbing}
\end{algorithm}

\end{description}

\begin{algorithm}[t]
\KwIn{$i$ - the behavior, $\mathbf{b}[i]$ - number of seeds required for the $i$th behavior, $\mathbf{S}[i]$ - the set of already selected seeds for the $i$th behavior, $V'$ - the remaining population of nodes to choose new seeds from}
\KwOut{$T_i $ - the set of $\mathbf{b}[i]$ newly selected seeds}
Let $e_i (v):=1$ for all $v\in V \setminus \mathbf{S}[i]$ and $i \in \{1,\ldots,k\}$\;
\For{each $v \in V \setminus S[i]$}{
   \For{each neighbor $u$ of $v$ s.t. $u \in V\setminus \mathbf{S}[i]$}{
      \If{$r(u)\ge c_i $}{
         $e_i (v):=e_i (v)+\frac{1}{|N(u)|}$\;
      }
   }
}
Let $T_i :=\phi$\;
\For{$j=1$ to $\mathbf{b}[i]$}{
    Select $u:=\argmax_{v}\{e_i (v)|v\in V' \setminus T_i \}$ \;
    $T_i :=T_i  \cup \{u\}$\;
    \For{each neighbor $v$ of $u$ in $V' \setminus T_i $}{
       $e_i (v):=e_i (v) - \frac{1}{|N(u)|}$\;
    }
}
\caption{Core-Hill-Climbing}
\label{algo:Core-One-Step-Hill-Climbing}
\end{algorithm}

\subsubsection{Expected Immediate Adoption}
Notice that the randomized behavior adoption process at a node in the network does not depend on any other node once the behaviors adopted by its neighbors are known. So for a seed set $\mathbf{S}$ we can compute the expected number of nodes with behavior $i$ in the next time step in the following manner --- for any node $v$ that is a neighbor of a seed node $u$ with behavior $i$, we compute the probability that in the next time step it will adopt behavior $i$, and then we sum up the probabilities for all such nodes $v$ in the network. Appendix \ref{sec:ia-comp} explains how to compute this probability with a detailed example. For a given seed set $\mathbf{S}$ and a behavior $i$, we define \textit{Expected Immediate Adoption(EIA)}, $\mathrm{IA}_i (\mathbf{S})$, as this expected number of nodes with behavior $i$ after one time step. Although the exact computation of \textit{total} number of adoptions at the completion of the behavior diffusion process is \#P-hard for the LT model \cite{Chen10}, the exact computation of adoptions after exactly one time step is a tractable problem (Appendix \ref{sec:ia-comp}). If we assume that higher values of $\mathrm{IA}_i (\mathbf{S})$ would result in higher values of expected number of adoptions of behavior $i$ at the completion of the diffusion process, then we can use $\mathrm{IA}_i (\mathbf{S})$ as a metric for determining the initial seed set. Instead of one step expected adoption numbers, one can use expected number of adoptions after two or more steps to better predict the final adoption. However, finding closed form expression is hard, and one would need to use a two-stage simulation to make the calculation.

We build up $\mathbf{S}$ incrementally adding one seed at a time based on the EIA value.  At each iteration we choose the node that maximizes the marginal increase in the EIA value to add to the seed set. See Algorithm \ref{algo:inc-IA-S-T} for detailed description.

\begin{algorithm}[h]
\KwIn{$G:=(V,E)$ - the social network, $\mathbf{b}$ - a vector of size $k$ containing number of required seeds for each of the behaviors} 
\KwOut{$\mathbf{S}$ - a vector of size $k$ containing seed sets of required size for all the behaviors} 
Let $V':=V$ and $\mathbf{S:=}\boldsymbol{\phi}$\; 
\Repeat{$\mathbf{b}=\mathbf{0}$}{     
    \For{each behavior $i$}{        
        Let $(u_i , s_i ):=$Find-Next-Seed-IA($i, \mathbf{b}[i], \mathbf{S}, V'$) \; \nllabel{line:next-best-seed-IA}
       }     
       Let $i_{max}:=\argmax_{i\in \{1,\ldots,k\}}s_i $ \; 
       Let $v:=u_{i_{max}}$ \;   
       Set $V':=V'\setminus \{v\}$ \;          
       Set $\mathbf{S}[i_{max}] := \mathbf{S}[i_{max}] \cup \{v\}$ and $\mathbf{b}[i_{max}] := \mathbf{b}[i_{max}] - 1$\;  
       \If{$r(v) \le c_{i_{max}}$} {
           Set $r(v):=c_{i_{max}}$\;
       } 
   } 
   \caption{Incremental Expected Immediate Adoption Based Heuristic} 
   \label{algo:inc-IA-S-T}
\end{algorithm} 

\begin{algorithm}[h]
\KwIn{$i$ - the behavior, $\mathbf{b}[i]$ - number of seeds required for the $i$th behavior, $\mathbf{S}$ - the set of already selected seeds for all the behaviors, $V'$ - the remaining population of nodes to choose new seeds from} 
\KwOut{$(u,s)$ - if $\mathbf{b}[i]$ is not zero then a tuple consisting of the best choice of next seed from the population $V'$ for the $i$th behavior, given the already selected seedset $\mathbf{S}$ and its corresponding Expected Immediate Adoption value}

\If {$\mathbf{b}[i] = 0$} {
    Return (`nobody',$0$)
}
Let $\mathbf{s}$ be a vector indexed by the set $V'$, and $\mathbf{s}=\mathbf{0}$ \; 
\For {$v \in V'$} {
	Let $\mathbf{S'} := \mathbf{S}$ \;
	Set $\mathbf{S'}[i] := \mathbf{S'}[i] \cup \{v\}$ \;
    Set $\mathbf{s}[v] :=$ Compute-IA($\mathbf{S'}$) \; \nllabel{line:compute-IA}
}
Select $u:=\argmax_{v}\{\mathbf{s}[v]|v \in V'\}$ \;    
Return $(u,\mathbf{s}[u])$;
\caption{Find-Next-Seed-IA heuristic; selects the node that gives maximum marginal increase of the Immediate Adoption value} 
\label{algo:inc-Core-Hill-Climbing} 
\end{algorithm}

\subsubsection{Greedy Approximation (KKT)}
\citeN{Kempe03} presents a greedy approximation algorithm with approximation guarantee of $63$\% for the LT model and single behavior case. In section \ref{sub-sec:approx-algo} we have shown that a modified version of this algorithm provides us with the same approximation guarantee for the simplified Sticky Model.  In particular, we showed that our algorithm under the total participation diffusion metric was sub-modular and hence we shall use this algorithm to select seeds. However, since we did not show that behavior diffusion to be sub-modular and hence prove the approximation guarantee under the the more general case of total adoption or under the resource utilization metric, using our algorithm may be less than optimal. Hence we denote our algorithm as a heuristic---we call it the KKT heuristic after the authors of~\cite{Kempe03}. note that due to the high computational cost involved in simulation this algorithm is not scalable to large size networks.

\subsection{Simulation Experiments} \label{sub-sec:sim-exp}
In this section we describe different simulation experiments and compare the effectiveness of our proposed heuristics for the seed selection problem. We have implemented the multiple behavior diffusion model described in Section~\ref{sub-sec:model} and the heuristics discussed in Section~\ref{sec:heuristics} in the NetLogo Programming environment~\cite{Wilensky99}. In the following experiments we have assumed that we want to spread three behaviors ${b_1, b_2, b_3}$ with costs $c_1=0.2$, $c_2=0.5$ and $c_3=0.7$. We have assumed that behavior utility is proportional to cost. Hence our nominal utility values for the corresponding behaviors are $u_1=0.2$, $u_2=0.5$ and $u_3=0.7$. Finally, we assume that individuals' resources are independent and identically distributed i.e the resource $r(v)$ is uniformly distributed random variable $U(0,1)$ for all $v \in V$.

\subsubsection{Network Topologies}
We have used synthetic networks as well as a large real-world network for our experiments.  We synthesize network topologies through three social network generation models: preferential attachment \cite{barabasi99}; Small-world \cite{watts98} and spatially clustered \cite{Stonedahl08}. All the synthetic networks have $500$ nodes. In the preferential attachment network each new coming node adds one link to one of the existing nodes according to the in-degree distribution. The small world network formation starts with a regular ring lattice where each node is connected to two adjacent nodes in the circular order. In the rewiring stage each edge is rewired with probability $p=0.2$. In the spatially clustered network average node degree is set to $10$. The three synthetic networks exhibit all the important properties---low effective diameter, power law degree distribution and high clustering---found in real world social networks. The real world data set is the ca-GrQc collaboration network form the SNAP network database \cite{Leskovec07a}. It is a collaboration network amongst authors who submitted their papers to the General Relativity and Quantum Cosmology category of e-print arXiv.org database. This network has $5242$ nodes and $28980$ edges.

The network types are abbreviated in the tables with experimental results as follows: PA (Preferential Attachment);  SW (Small World); SC (Spatially Clustered); QC (the ca-GrQc quantum cosmology collaboration network form the SNAP network database).
 
\subsubsection{Empirical Evaluation}
In this section we compare the seven seed selection heuristics described in Section~\ref{subsec:heuristics} for different network topologies.  For the seed selection experiments, we fix the behavior distribution over the seeds: the behaviors are assumed to be uniformly distributed over the seeds. We use a specific fraction $\alpha$ of the population as seeds. In this experiment, we have used $\alpha=0.1$ in line with prior work ~\cite{Watts07}. This means that for synthetic networks, we use $b=51$ seeds\footnote{the number of seeds is a multiple of 3, since we have 3 test behaviors}, and $b=501$ for the real-world network. All the results discussed in this section are for the \textbf{S-T} variant of the algorithm\footnote{In Appendix \ref{sec:t-nt-comp} we present the result of comparison between the different variants}. As a reminder the \textbf{S-T} variant is the case when each seed node adopts a single behavior with top-up. Full description of the notations can be found in Section~\ref{subsec:notation}.

There are two sources of randomness in the synthetic network generation models:  behavior adoption thresholds at each individual for each behavior and network topology. Since each aspect is independent of the other, we have conducted two different types of simulations. In the first, we pick an arbitrary topology and vary individual thresholds over the different simulation runs. We term this as \textit{threshold average}. In the second type of simulation, we fix the individual thresholds, obtained from the uniform distribution, and vary the topologies over the simulations. We term this as \textit{network average}. Notice that the real-world dataset---ca-GrQc network---has a fixed topology and hence only one type of randomness: variation of the individual thresholds.  We use $5000$ independent runs of the diffusion process to obtain stable estimates for both threshold and network types of simulations.

Interestingly, simulation under both network and threshold average yield identical results (see Table~\ref{tab:max-util}). While the reader can find a formal proof of equivalence between the network average and threshold average simulations in Appendix~\ref{app:equiv}, we present an informal argument here. The case for $k-$regular graphs is easy to see. No matter the topology, every node will experience the same distribution of thresholds as any other node as the number of simulations tends to infinity. In the more general case, consider the network thresholds case when topologies are varied under fixed (but randomly assigned) thresholds. Since every topology is the output of a random graph generation process (e.g. small world), every topology follows a characteristic degree distribution (depending on the graph generation parameters). Thus any node with a fixed threshold has a finite probability of having all the possible different degrees. Since the thresholds are chosen uniformly at random, across all nodes, every threshold value will ``experience" different node degrees (i.e. number of neighbors) over the course of the entire simulation. A similar argument follows for the threshold average case when the topology is random but fixed and when the thresholds are chosen randomly for every simulation run. In the remainder of the paper, we shall only show threshold averages for the sake of definiteness.

The eight heuristics are abbreviated in the experimental results tables as follows: Random (H1); Na\"ive Degree---No Top-up (H2); Na\"ive Degree---Knapsack (H3); Na\"ive Degree---Top-up (H4); Degree and Resource Ranked (H5), Constrained Social Influence Weight---Ranked (IWR), Constrained Social Influence Weight---Max Margin (IWM), Expected Immediate Adoption (EIA).

\begin{table}[htb]
\footnotesize
\centering
    \caption{Maximum Possible Resource Utilization of different network types.  Each node solves the knapsack problem and selects optimal behaviors. Then, we diffuse the behaviors. We are reporting the equilibrium values under two conditions: we fix the thresholds and vary topology (Network Average); we fix a random topology and vary thresholds (Threshold Average). Notice that the the quantum physics collaborative dataset, we cannot report a network average since the topology is fixed.  Maximum resource utilization occurs when the number of seeds is equal to N---each node is a seed. In this case, it is easy to show that resource utilization is 0.78 for the three behavior case with specific costs. Notice below that resource utilization is less than this number, since each node will adjust to its social signal. }\label{tab:max-util}
    \begin{tabular}{lcc} \toprule
        Network & Threshold Average & Network Average \\ \midrule
        (PA) Preferential Attachment 	& 0.71 & 0.71	\\
        (SW) Small World 			& 0.72 & 0.72 	\\
        (SC) Spatially Clustered 		& 0.73 & 0.73	\\
        (QC) Quantum Cosmology 		& 0.73 & N/A 	\\ \bottomrule
    \end{tabular}
    \end{table}

Since seed selection sub-problems P1 and P2 are NP-complete (ref. Section~\ref{sec:heuristics}), determining the maximum possible utilization or total participation in the network for the given value of $b$ under uniform behavior distribution is computationally intractable. However, we can estimate the value of maximum possible utilization in the network if we assume that $b=N$, the case when each network node is a seed. First the nodes in the network adopt the subset of behaviors that maximizes their payoff subject to the resource constraint.  Then we let the diffusion process run till the network reaches equilibrium. The expected value of the resource utilization at this point will upper bound of resource utilization in that network and enables comparison with our heuristics. Table \ref{tab:max-util} provides the value of this maximum possible utilization for different networks. Notice that for three behaviors with costs $c_1=0.2$, $c_2=0.5$ and $c_3=0.7$, it is straightforward to show that the maximum utilization will be bounded by the value 0.78, assuming that the thresholds are obtained from $U(0,1)$. The fact that the simulation results are slightly lower that 0.78 is because nodes will ``align'' with their neighbors over time due to the social influence.

Let us examine the resource utilization for our cost distribution $\bm{c} = (0.2, 0.5, 0.7)$ in a little more detail. We assume that an individual's resources are drawn from a uniform distribution $U(0,1)$. If we assume that every individual can adopt any subset of available behaviors provided they have the resources, independent of the behavior adoption by their neighbors, we can thus estimate the upper bound for resource utilization in the social network. Given the distribution of costs, not all individuals can fully utilize their resources. However, it is easy to see that individuals with resources exactly equal to one of the numbers from the set $\{0.2, 0.5, 0.7, 0.9 \}$ are able to fully utilize their available resources.  We call this the set of  \textit{full utilization points}. Now, an individual with resource say 0.4 can only adopt behavior with cost 0.2 and so on. Thus the expected resource utilization is: 

\begin{equation*}
\frac{\int_{0.2}^{0.5}0.2 \mathrm{d}r + \int_{0.5}^{0.7}0.5 \mathrm{d}r + \int_{0.7}^{0.9}0.7 \mathrm{d}r + \int_{0.9}^{1.0}0.9 \mathrm{d}r  }{\int_{0.0}^{1.0}r \mathrm{d}r} = 0.78
\end{equation*}

The denominator of the equation is the value of resource utilization when an individual has behaviors that exactly match available resources. The costs of the different behaviors completely specifies the set of full utilization points.  It is straightforward to prove the following lemma.

\begin{lemma}
Let there be $n$ utilization points $\mu_{1}, \mu_{2}, \dots, \mu_{n}$ with $0 < \mu_{1} < \mu_{2} < \dots, < \mu_{n} \leq 1.0$. Then the maximum resource utilization is $2(\mu_{1}\mu_{2} + \cdots + \mu_{n-1}\mu_{n} + \mu_{n}) - 2(\mu_{1}^{2}+ \cdots + \mu_{n}^{2})$.
\end{lemma}
\begin{proof}
The expected value of the resource utilization is: 
\begin{align}
& \hphantom{{} = 1} \frac{\int_{\mu_{1}}^{\mu_{2}} \mu_{1} \mathrm{d}r + \int_{\mu_{2}}^{\mu_{3}} \mu_{2} \mathrm{d}r + \cdots + \int_{\mu_{n}}^{1.0} \mu_{n} \mathrm{d}r} {\int_{0.0}^{1.0}r \mathrm{d}r} \notag\\
& = 2 (\mu_{1}(\mu_{2} - \mu_{1}) + \cdots + \mu_{n-1}(\mu_{n} - \mu_{n-1}) + + \mu_{n}(1.0 - \mu_{n}) )\notag\\
& = 2 (\mu_{1}\mu_{2} + \cdots + \mu_{n-1}\mu_{n} + \mu_{n}) - 2(\mu_{1}^{2} + \mu_{2}^{2}+ \cdots + \mu_{n}^{2}) \notag.
\end{align}
\end{proof}

By setting partial derivatives with respect to the utilization points of the maximum resource utilization function to 0, it is straightforward to show that the utilization points must be uniformly distributed. That is, $\mu_{i} = i / (n+1)$ and the corresponding maximum resource utilization value is $2/(n+1) \times \sum_{i=1}^{n}\mu_{i} = n / (n+1)$. Since $\mu_{1}$ always represents the lowest cost behavior, by simply introducing a new lower cost behavior with cost $\mu_{1}/2$, we increase the number of utilization points from $n \rightarrow 2n+1$, and increasing the relative utilization by $1/2n$.  As a specific example with three utilization points, $\bm{\mu} = (0.25, 0.5, 075)$ the costs for the corresponding two behaviors are $\bm{c} = (0.25, 0.5)$, with maximum resource utilization of $3/4$. When we introduce a new behavior with cost equal to half the cost of the lowest cost behavior (i.e. 0.125), the number of utilization points jumps to 7, thereby increasing the utilization to $7/8$ and increasing relative utilization by $1/6$ or 16.7\%.   

\begin{table}[htb]\footnotesize
\centering
    \caption{Resource Utilization under Threshold  Average. Both heuristic variants of the Influence Weight give excellent results. The differences between the heuristics for the same type of average are statistically significant. * result is obtained with 50 independent runs.}\label{tab:seed-selection-util}
    \begin{tabular}{lcccc} \toprule
        Seed Selection Heuristics & PA & SW& SC& QC\\ \midrule
        Random 								& 0.12  & 0.15 & 0.16 & 0.14 \\
        Na\"ive Degree---No Top-up  				& 0.22 & 0.16  & 0.16 & 0.18 \\
    	Na\"ive Degree---Knapsack 					& 0.28 & 0.17  & 0.16 & 0.18 \\
        Na\"ive Degree---Top-up 					& 0.32 & 0.17  & 0.17 & 0.19 \\
        Degree and Resource Ranked				& 0.35 & 0.21  & 0.18 & 0.20 \\
        Constrained Social Influence Weight---Ranked 	& \textbf{0.37} & \textbf{0.21} & \textbf{0.20}  & \textbf{0.28} \\
        Constrained Social Influence Weight---Max Margin 	& \textbf{0.37} & \textbf{0.22} & \textbf{0.21}  & \textbf{0.29} \\ 
        Expected Immediate Adoption 				& 0.34 &  0.22 &  0.21& 0.27* \\ \bottomrule
    \end{tabular}
\end{table}

Table \ref{tab:seed-selection-util} shows the estimated resource utilization of different networks for threshold and network average simulations for each of the eight seed selection heuristics. The two Constrained Social Influence Weight  heuristics (Ranked and Max Margin) show the highest expected utilization. The differences between the influence weight heuristics and the other heuristics are statistically significant ($p < 0.01$).

\begin{table}[htb]\footnotesize
    \centering
    \caption{Total Participation / Total Adoption under Threshold average as a percentage of the network size for four different types of networks---Preferential Attachment (PA); Small World (SW); Spatially Clustered (SC); Quantum Cosmology (QC). Both versions of the Constrained Social Influence Weight heuristics and the Expected Immediate Adoption heuristic give excellent results. The differences between the heuristics for the same type of average are statistically significant. * result is obtained with 50 independent runs}\label{tab:seed-selection-adoption}
    \begin{tabular}{lcccc} \toprule
        Seed Selection Heuristics & PA & SW& SC& QC\\ \midrule
        Random 								& 14.0 / 14.1 & 17.7 / 18.0 & 20.1 / 20.6 & 17.3 / 17.6 \\
        Na\"ive Degree---No Top-up  				& 27.2 / 28.3 & 19.7 / 20.3 & 20.3 / 20.6 & 20.2 / 22.4 \\
    	Na\"ive Degree---Knapsack 					& 31.6 / 35.7 & 19.1 / 21.3 & 18.6 / 20.7 & 22.2 / 23.4 \\
        Na\"ive Degree---Top-up 					& 37.5 / 38.0 & 21.3 / 21.9 & 21.1 / 21.8 & 22.6 / 23.8 \\
        Degree and Resource Ranked				& 41.3 / 41.6 & 24.7 / 25.4 & 22.9 / 23.4 & 22.5 / 23.3 \\
        Constrained Social Influence Weight-Ranked 	& 45.0 / 45.2 & 25.0 / 25.3 & 25.8 / 26.3 & 34.8 / 35.9 \\
        Constrained Social Influence Weight-Max Margin 	& 44.0 / 44.4 & 25.9 / 26.4 & 25.7 / 26.3 & 35.5 / 36.6 \\
        Expected Immediate Adoption 				& 51.3 / 52.1 & 26.9 / 27.5 & 27.6 / 28.4 & 37.4 / 38.8*\\ \bottomrule
    \end{tabular}
\end{table}

Table \ref{tab:seed-selection-adoption} presents the Total Participation and Total Adoption under threshold average condition for all the eight heuristics. The Expected Immediate Adoption based heuristic (EIA) shows the best result. The results remain qualitatively unchanged in the network average case.

\begin{table}[htb]\footnotesize
    \centering
    \caption{Total Participation / Total Adoption under different networks as a percentage of the network size for three different types of networks---Preferential Attachment (PA); Small World (SW); Spatially Clustered (SC).  The heuristics Constrained Social Influence Weight-Ranked, Constrained Social Influence Weight-Max Margin and the Expected Immediate Adoption give results quite close to the greedy approximation algorithm.}\label{tab:comp-to-approx}
    \begin{tabular}{lccc} \toprule
        Heuristics & Preferential Attachment & Small World & Spatially Clustered \\ \midrule
        Greedy Algorithm & 43.7 / 44.5 & 26.2 / 26.4 & 27.3 / 27.3  \\
        \textbf{Constrained Social Influence Weight-Ranked} & 43.9 / 44.5 & 22.9 / 23.6 & 24.6 / 25.1 \\
        \textbf{Constrained Social Influence Weight-Max Margin} & 33.3 / 33.4 & 22.9 / 23.6 & 23.5 / 24.1 \\ 
        \textbf{Expected Immediate Adoption} & 43.9 / 44.5 & 23.6 / 24.5 & 23.6 / 24.2 \\ \bottomrule
    \end{tabular}
\end{table}

We compare the performance of the two Influence Weight  based heuristics---Constrained Social Influence Weight-Ranked and Constrained Social Influence Weight-Max Margin---and the Expected Immediate Adoption based heuristic against the greedy approximation algorithm. Due to a large computational cost it is infeasible to run the greedy algorithm on the large networks used in experiments thus far.  So we create synthetic networks of size $100$ nodes with number of seeds $b=9$ with $3$ seeds allocated to each behavior, for the purpose of comparison via Preferential Attachment, Small World and Spatially Clustered network generators.  Our estimates are obtained by running the model $5000$ times on each network. Table \ref{tab:comp-to-approx} presents the results of the comparison for the total participation metric. Notice that for the Preferential Attachment network heuristics Constrained Social Influence Weight-Ranked and the Expected Immediate Adoption perform even better than the greedy approximation algorithm. This is not surprising since the greedy algorithm may fail to obtain the optimal solution in isolated cases. In the next section, we discuss how to distribute behaviors over the seeds.

\section{Assigning Seeds Optimal Behaviors} \label{sec:dist}
In this section we discuss the behavior distribution problem---what set of behaviors should each seed node adopt? We will first formally introduce the problem as an optimization problem. Next we will discuss different strategies for distributing the behaviors over the seed set. Then we will compare these different strategies through simulation experiments and identify the pros and the cons of each strategy.

The behaviors adopted by the set of seed nodes have different implications on the metrics---total participation, total adoption and resource utilization (see Section~\ref{subsec:measure} for a definition of these measures). If all nodes adopted the least costly behavior, for example, we would expect total participation to increase, but low resource utilization. The converse would be true in the case when seed nodes are chosen in such a way that all adopt the most expensive behavior. However, if we want to strike a balance between different behaviors such that all the behaviors are represented in the population, then we will have to distribute all the behaviors over the seed set according to some ratio. Here we formalize this scenario as an optimization problem.

\begin{description}
\item[P3: Optimum Behavior Distribution over the Seed Set:] Given that we are to pick $b$ seeds and there is a lower bound on the number of adoptions of the lowest cost behavior $s_{min}$, 
identify the optimum distribution of behaviors over the seed set and the optimum set of $b$ seeds that will maximize the resource utilization while maintaining expected spread of $s_{min}$ for the lowest cost behavior.
\end{description}

It should be noted that in the case of multiple behavior diffusion metrics like resource utilization, total participation and total adoption depends not only on the choice of the seeds but also on the distribution of the different behaviors in the chosen seed set. In the results that follow, we assume that each seed node is assigned a single behavior.  Note that this assignment doesn't preclude the seeds as well other nodes from adopting more than one behavior. We test following five different distributions of the behaviors in the seed set. In the \textit{highest cost behavior only} distribution we allocate all the seeds to the highest cost behavior and none to the other behaviors. In the \textit{proportional to cost} distribution the behaviors are distributed over the seeds in the ratio of their costs. That is, higher cost behaviors have more seeds allocated to them. The justification for this heuristic is that higher cost behaviors will diffuse less and hence if we would like to see higher resource utilization, we should set more seeds to adopt the higher cost behaviors. \textit{Uniform} distribution assigns the seeds a behavior chosen uniformly at random all the behaviors. In the \textit{Inversely proportional to cost} behavior distribution behaviors are distributed over the seeds in the inverse ratio of their costs. So the highest cost behavior gets the lowest number of seeds and the lowest cost behavior gets the highest number of seeds. Finally, in the \textit{lowest cost behavior only} distribution all the seeds are assigned to the lowest cost behavior and no seeds are given to the other behavior.  Once we have identified how many seed nodes are needed for each behavior, using the aforementioned strategies, we identify  the optimal seeds for that behavior via the Constrained Social Influence Weight---Max Margin heuristic.

\subsection{Experimental Evaluation}

In this section we investigate the effects of the different behavior distribution heuristics across the initial seed set described in the previous section. For this simulation, we use Constrained Social Influence Weight-Max Margin heuristic, since it is one of the best performing seed selection heuristics (see Table~\ref{tab:seed-selection-util}).  Adopting the strategy in~\citeN{Watts07}, we designate the fraction of seeds designated to be early adopters to be $\alpha=0.1$. This means that we have $b=51$ for the synthetic networks ($N=500$) and $b=501$ for the quantum physics collaboration network.  As before, we compute the metrics under the threshold average and the network average simulations.

In the tables in this section, we shall use the following notation: Low (All seeds are assigned Lowest Cost Behavior); Inv. (the seeds are allocated behavior in Inverse proportion to behavior cost; Unif. (the behaviors are distributed Uniformly at random); Prop. (the behaviors are distributed Proportional to behavior cost); High ( all seeds are allocated the Highest cost behavior).

\begin{table}[htb]\footnotesize
\centering
\caption{Resource Utilization under Threshold Average.  Among the behavior distribution heuristics, assigning every seed the lowest (highest) cost behavior results in the lowest (highest) utilization. Assigning seeds proportional to cost, works as well as the assigning everyone the highest cost behavior.  }\label{tab:behav-util-threshold}
    \begin{tabular}{lcccc} \toprule
        Heuristics 					& PA 		& SW 	& SC 	& QC \\ \midrule
        Lowest cost 				& 0.23   	& 0.14 	& 0.15	& 0.18 \\
        Inversely proportional  		& 0.33   	& 0.20 	& 0.20	& 0.27 \\
        Uniform distribution 			& 0.37	& 0.22 	& 0.21	& 0.29 \\
        \textbf{Proportional to Cost} 	& 0.38    	& 0.24 	& 0.22	& 0.31 \\
        \textbf{Highest Cost} 			& 0.38  	& 0.24	& 0.24	& 0.31 \\ \bottomrule  
    \end{tabular}
\end{table}

\begin{table}[htb]\footnotesize
\centering
  \caption{Total Participation / Total adoption under Threshold for different behavior distributions over seeds. Seeds are chosen under Constrained Social Influence Weight---Max Margin heuristic. Notice that when all the seeds are the same behavior (Low, High), the number of unique participants and adoptions are identical. }\label{tab:behav-part-network}
    \begin{tabular}{lcccc} \toprule
        Heuristics 				& PA 				& SW 			& SC \\ \midrule
       Lowest cost  			& 291.12 / 291.12	& 166.26 / 166.26	& 178.78 / 178.78 \\
       Inversely proportional  		& 250.91 / 254.52	& 146.36 / 150.09	& 149.61 / 154.58\\
       \textbf{Uniform distribution} 	& 234.00 / 236.46	& 133.38 / 136.04	& 132.27 / 135.77\\
       \textbf{Proportional to Cost} & 209.66 / 210.98	& 118.65 / 119.98	& 113.29 / 114.91\\
       Highest Cost 			& 144.49 / 144.49  	& 93.79 / 93.79 		& 86.01 / 86.01\\ \bottomrule        
    \end{tabular}
 \end{table}{\tiny }

Table~\ref{tab:behav-util-threshold} shows the resource utilization in different networks for the threshold average and the network average simulations. We see that when each seed is allocated the same lowest (highest) cost behavior, the utilization is lowest (highest). This is unsurprising as we should expect high utilization to occur when we have high cost behaviors in the network. In Table~\ref{tab:behav-part-network}, we show the difference between the number of unique participants and the total number of behavior adoptions. We have omitted the simulations for the network average case, due to space limitations.  Those simulations are qualitatively similar to Table~\ref{tab:behav-part-network}. Notice that when all the seeds have either the same lowest cost or highest cost behavior assigned to all of them, there is unsurprisingly no difference between the total number participants and the total number of unique adoptions. As Table~\ref{tab:behav-part-network} shows, change to the behavior distribution over the seeds alters the unique number of participants as well as the total adoption. Therefore the seed distributions need to be chosen with care, the appropriate metric in mind. Both uniform and proportional to cost behavior distribution methods seem to hit a sweet spot between utilization and behavior diversity.

\section{Discussion and Open Issues}\label{sec:disc}
One of the main motivations of the present work was to develop a realistic model of the behavior diffusion process. There are many ways in which our work can be extended. Here we discuss about a few such possible extensions.
    
Our present model does not consider the role of behavioral inertia in the diffusion process. Often people are hesitant of adopting new behaviors because they cannot free their resources from practicing an old behavior which possibly has less value. This can be modeled in our framework by introducing an additional benefit for the already adopted behaviors. Another technique would be to introduce epidemic models such as SIRS to better model long-term behavioral adoption. 

In a network, we receive social signals from our friends, but there is noise because we miss messages and or we check them late. In modeling the behavior adoption problem, we have ignored the role of constraints in how they affect the production and consumption of messages from peers. Explicit consideration of the cost of social signaling would not only make the model more realistic and provide better bounds on the maximal resource utilization of the networks resources.

The study of social diffusion and information contagion has met with its fair share of criticism. \citeN{Aral09} argues that in their observational study more than $50\%$ of the perceived behavior contagion can be attributed to homophily instead of social influence. However \citeN{shalizi11} have shown that homophily and social influence are generically confounded in social diffusion processes and it in general not easy to distinguish between the two effects. We have tried to take this observation into account while developing our model (see Section~\ref{sub-sec:model}). In our model, the diffusion process is not exclusively driven by the social influence effects but an individual's intrinsic characteristics including interest in the behavior as well as resource constraints.  

\citeN{Goel12} has provided an important critique of research that focuses on modeling large-scale adoption through contagion like models. Their research shows that much of the empirically observed cascades on a variety of networks is small---in most cases, the cascade stops within one degree of the initial adopting seed. However, we note that large cascades do occur but are rare---the ``Spanish Flu'' of 1918-19 infected 500M people, a third of the world's population at the time~\cite{Johnson2002,Taubenberger2006}. The main difference between~\citeN{Goel12} and our own work is that while there is less empirical evidence for large cascades, we are interested in \textit{engineering} large cascades through careful seed selection. We conjecture that large cascades are less likely to occur in ``natural'' cascades because the chance that all the influentials adopt the behavior nearly simultaneously to cause a large cascade would be rare.

\section{Conclusions}\label{sec:conclusion}
In this paper we have considered the problem of seed selection to maximize resource utilization for multiple behavior diffusion processes. This problem has implications in a varied set scenarios, ranging from viral marketing campaign to mass adoption of sustainable behaviors and public health. We have considered a social network where individuals are constrained by available resources for adoption of new behaviors. Our work is the first of its kind, to the best of our knowledge to study the influence of individual resource constraints on multiple, costly behavior adoption. Mindful of the confound between homophily and structural effects, individuals in our model respond to the social influence as well as the intrinsic utility of a spreading behavior. We have shown that the core optimization problems are NP-complete and provided novel heuristics for solving them. We have tested our heuristics against the random and na\"ive methods and have shown that our heuristics perform very well. We have also shown that depending on the objective of the campaign, there are different strategies for distributing the behaviors over the initial seed set that result in qualitatively different outcomes. Some of the open issues include the use of epidemic models for modeling long-term behavior adoption and incorporating the idea of noisy social signals in modeling behavior adoption.

\newpage
\appendix
\section*{Appendix}
\section{Computation of Immediate Adoption Probability}
\label{sec:ia-comp}
In this section we discuss an example of how the different immediate behavior adoption probabilities for a node are computed. This computation depends on whether all the thresholds of a node have the same random value (\emph{matched threshold}) or independent and uniformly distributed random values (\emph{different threshold}). Although all the results in this paper are for the \emph{different threhsold} model, here we present examples for both cases for the sake of completeness. Suppose a vertex $v$ has $8$ neighbors. According to our threshold model each of its neighbors exerts an influence of $0.125$ on it. Suppose $v$ is already a seed for behavior \emph{A}; moreover it has $2$ neighbors with behaviors \emph{B}, and $3$ neighbors with behavior \emph{C}. We are interested in computing the probabilities that it will adopt each of the three behaviors in the next time step. 

\subsection{Matched Threshold:}
In this case, for any vertex the thresholds for all the three behaviors will be the same, but it will be assigned independently of other nodes and uniformly at random from the interval $[0,1]$. So if $v$'s threshold is in the interval $[0,0.25]$, then $v$ will consider both behaviors B and C together with A for adoption. Our payoff maximizing behavior adoption process dictates that it will adopt a subset of \emph{A,B} and \emph{C} that will provide maximum combined payoff subject to the resource constraint of the node. This adoption decision process is equivalent to solving a knapsack problem. We will solve the knapsack problem and decide which behaviors out of the three behaviors - A, B and C - will be adopted. Any such behavior will be adopted with probability $0.25.$ 

If $v$'s threshold is in the interval $(0.25,0.375]$ then $v$ will only consider behavior C together with behavior A for adoption. Again after solving knapsack problem and deciding which behaviors to adopt out of A and C, it will adopt any such behavior with probability $0.125$. 

At last if $v$'s threshold is in the interval $(0.375,1]$ then it will definitely adopt behavior A - the probability of which is $1-0.375=0.625$. In the worst case the complexity of this probability computation process for each node is linear in the number of behaviors.

\subsection{Different Threshold:}
In the \emph{different threshold} case, for each vertex the thresholds are assigned independently and uniformly at random from the interval $[0,1]$. So in this case we need to consider all possible combinations of behaviors B and C together with A (which will always be considered) and work out the individual probabilities. The worst case computational complexity of this process for each node will be exponential in the number of behaviors. In our example we need to consider the following cases:
\begin{enumerate}
\item[i)] B and C together with A; any behavior selected by the knapsack algorithm will be adopted with probability $0.25\times0.375=0.09375$.
\item[ii)] B together with A; any behavior selected by the knapsack algorithm will be adopted with probability $0.25\times(1-0.375)=0.15625$.
\item[iii)] C together with A; any behavior selected by the knapsack algorithm will be adopted with probability $(1-0.25)\times0.375=0.28125$.
\item[iv)] Only A; A will be adopted with probability $(1-0.25)\times(1-0.375)=0.46875$.
\end{enumerate}

As this section illustrates, computation of this metric is considerably more complex and costly in comparison to node degree and influence weight based heuristics.

\section{Variants of Seed Selection Algorithm} \label{sec:t-nt-comp}
Table \ref{tab:t-nt-comp} presents the Total Participation and Total Adoption values for the different variants of the KKT seed selection algorithm and IA based seed selection heuristic. \textbf{T} --- \emph{topped up} versions provide better spread than the \textbf{NT} --- \emph{no top up} versions which is expected since more resource is required for starting the diffusion in the \textbf{T} version. However for the same type of top up regime there is not much difference between the \textbf{S} (\emph{single} behavior per seed) and \textbf{M} (\emph{multiple} behaviros per seed) version. If we consider exact algorithms instead of heuristics and approximation algorithms, then it is easy to see that \textbf{S} version can never produce a result that is better than the \textbf{M} version, since solution for \textbf{S} version is also a valid solution for \textbf{M} version. This fact accounts for the absence of any real difference between the \textbf{S} and \textbf{T} versions in the case of the heuristic and the approximate algorithm. 
\begin{table}[htb]\footnotesize
    \centering
    \caption{Total Participation / Total Adoption under different networks as \% of the network size. S and M variants give almost identical results with T variants exceeding NT variants.}\label{tab:t-nt-comp}
    \begin{tabular}{cccc} \toprule
        Heuristics & PA & SW & SC \\ \midrule
        KKT-S-T & 43.7 / 44.5 & 26.2 / 26.4 & 27.3 / 27.3  \\
        H8-S-T & 43.9 / 44.5 & 23.6 / 24.5 & 23.6 / 24.2 \\ 
        \midrule
        KKT-S-NT & 39.5 / 39.5 & 21.7 / 22.0 & 22.0 / 22.5  \\
        H8-S-NT & 39.51 / 39.8 & 22.7 / 23.2 & 20.0 / 20.5 \\ 
        \midrule
        KKT-M-T & 43.7 / 44.5 & 26.2 / 26.4 & 27.1 / 27.1  \\
        H8-M-T & 39.0 / 45.8 & 22.8 / 23.5 & 21.9 / 22.6 \\ 
        \midrule
        KKT-M-NT & 39.5 / 39.5 & 21.7 / 22.0 & 22.4 / 23.0  \\
        H8-M-NT & 39.5 / 43.3 & 22.7 / 23.2 & 19.7 / 21.1 \\ 
        \bottomrule
        
    \end{tabular}
\end{table}

\section{Equivalence between the Threshold and Network Average Cases}
\label{app:equiv}
In table \ref{tab:seed-selection-util} we have seen that the resource utilization values under threshold and network average conditions are almost identical. In this section we will investigate the relationship between these two type of averages. First we will show an exact relation for the regular networks. This special case will provide us with helpful insights for analyzing the more general cases.

Suppose we have $n$ nodes with fixed resource distribution. Each node will have a fixed in-degree $\rho$. Each node selects $\rho$ in-neighbors uniformly at random from the rest $n-1$ nodes. We assume that only in-neighbors can exert influence on a node. In the \textit{threshold average} (\textbf{TA}) case the nodes choose the in-neighbors at random at the beginning of the simulation and then at the start of each simulation run select the threshold values uniformly at random from the interval $[0,1]$. In the  \textit{network average} (\textbf{NA}) case each node chooses threshold values uniformly at random from the interval $[0,1]$ at the beginning of the simulation and then at the start of each simulation run it chooses its $\rho$ in-neighbors uniformly at random from the rest of the nodes. Both the processes start with a set $S$ of seeds for each of the $k$ behaviors. The diffusion process unfolds over time according to the Sticky multiple behavior diffusion process. We will show that $\sigma_{TA}(S)=\sigma_{NA}(S)$ by proving the following lemma:

\begin{lemma}
    For a given seed set $S$, the following two distributions over the sets of nodes are the same:
    \begin{enumerate}
        \item The distribution of probability over the active sets at the completion of the diffusion process in the \textbf{TA} case.
        \item The distribution of probability over the active sets at the completion of the diffusion process in the \textbf{NA} case.
    \end{enumerate}
\end{lemma}

\begin{proof}
    We prove the lemma by induction over the time step $t$. Clearly it is true at $t=0$. Let $S_i^{(t)}$ denote the set of nodes with behavior $i$ at the end of time step $t$, and $S^{(t)}:=\cup_i S_i^{(t)}$. For the \textbf{TA} case, suppose $v$ is a node that has not adopted any behavior at the end of time step $t$ and $\kappa(v)=\kappa \ne 0$. As before, the probability that $v$ will become active at the time step $t+1$, given that it was not active at the previous time step is -
    \begin{align*}
	    & 1-\prod_{i=1}^{\kappa}\left(1-\frac{\sum_{w\in S^{(t)}_{i}\setminus S^{(t-1)}_{i}}b_{v,w}}{1-\sum_{w\in S^{(t-1)}_{i}}b_{v,w}}\right) \\
	    = & 1-\prod_{i=1}^{\kappa}\left(1-\frac{\sum_{w\in S^{(t)}_{i}\setminus S^{(t-1)}_{i}}\frac{1}{\rho}}{1-\sum_{w\in S^{(t-1)}_{i}}\frac{1}{\rho}}\right) \\
	    = & 1-\prod_{i=1}^{\kappa}\left(1-\frac{ |S^{(t)}_{i}\setminus S^{(t-1)}_{i}|}{\rho-|S^{(t-1)}_{i}|}\right)
    \end{align*}
    For the \textbf{NA} case, again let $v$ be a node that is not active at time step $t$ with $\kappa(v)=\kappa \ne 0$. The probability that $v$ will become active at time step $t+1$, given that it was not active till the previous time step is given by -
    \begin{align*}
	    & 1-\prod_{i=1}^{\kappa}\left(1-\frac{\sum_{w\in S^{(t)}_{i}\setminus S^{(t-1)}_{i}}b_{v,w}}{1-\sum_{w\in S^{(t-1)}_{i}}b_{v,w}}\right) \\
	    = & 1-\prod_{i=1}^{\kappa}\left(1-\frac{\sum_{w\in S^{(t)}_{i}\setminus S^{(t-1)}_{i}}\frac{1}{\rho}}{1-\sum_{w\in S^{(t-1)}_{i}}\frac{1}{\rho}}\right) \\
	    = & 1-\prod_{i=1}^{\kappa}\left(1-\frac{ |S^{(t)}_{i}\setminus S^{(t-1)}_{i}|}{\rho-|S^{(t-1)}_{i}|}\right)
    \end{align*}
    Since the in-degree of every node is same, we get the same probability distribution over the active sets in both the cases.
\end{proof}

Consequently we obtain the result that the expected number of active nodes in both the \textbf{TA} and \textbf{NA} cases are the same for the networks with constant in-degree. In the general case when the networks do not have a constant degree for every node but the randomization over the network structure preserves a fixed degree distribution (as in the case of Power Law or  Spatially Clustered networks) we may obtain similar results. However the probability that a node becomes active in time step $t+1$, given that it was not active till time step $t$ would be calculated for a node $v$ with $\kappa(v)=\kappa \ne 0$ and degree $d \ne 0$. Assuming that the distribution over the values $d$ would be the same at the time step $t$ in  both the cases (notice that the distribution over the values $\kappa$ would be the same for both the cases since the initial distribution of node resources are the same), we will obtain similar results. Our experimental results show that this observations about the Sticky model carries over to the general model. In all of the simulation experiments we observe that the estimations of the expected values of the different metrics (total participation, total adoption, resource utilization etc.) for both the \textbf{TA} and \textbf{NA} cases are almost identical.

\newpage
\printbibliography 
\end{document}